\documentclass[preprint,12pt]{elsarticle}
\usepackage[utf8]{inputenc}
\usepackage{amsthm}
\usepackage{amsmath}
\usepackage{amsfonts}
\usepackage{amssymb}
\usepackage{graphicx}
\usepackage{epstopdf}
\usepackage{color}
\usepackage{multirow}
\usepackage{mathtools}
\usepackage{xcolor}
\usepackage{a4wide}
\usepackage{bm}
\usepackage{caption}
\usepackage{subcaption}
\usepackage[
  separate-uncertainty = true,
  multi-part-units = repeat
]{siunitx}
\usepackage{rotating}
\usepackage{verbatim}
\usepackage{lineno,hyperref}
\usepackage{tikz}
\tikzstyle{noeud} = [circle, draw, fill=white, inner sep=2pt]
\journal{Theoretical Computer Science}
\newtheorem{theorem}{Theorem}[section]
\newtheorem{lemma}{Lemma}[section]

\newdefinition{remark}{Remark}[section]
\newtheorem{observation}{Observation}[section]
\newdefinition{definition}{Definition} \newdefinition{example}{Example}

\usepackage[noend,linesnumbered,ruled,lined]{algorithm2e}
\SetAlFnt{\scriptsize}
\usepackage{amsfonts}
\usepackage{makecell}

\usepackage[usenames,dvipsnames]{pstricks}
\usepackage{pstricks-add}
\usepackage{epsfig}
\usepackage{pst-grad} 
\usepackage{pst-plot} 
\usepackage[space]{grffile} 
\usepackage{etoolbox} 
\makeatletter 
\patchcmd\Gread@eps{\@inputcheck#1 }{\@inputcheck"#1"\relax}{}{}
\makeatother
\usepackage{mathtools}

\usepackage{graphicx}
\usepackage{tabularx}
\usepackage{textcomp}
\usepackage{xcolor}
\usepackage[labelformat=simple]{subcaption}

\usepackage{booktabs}

\usepackage{pgf,tikz,pgfplots}
\pgfplotsset{compat=1.15}
\usepackage{mathrsfs}
\usepackage{hyperref}
\usepackage{multicol}
\SetAlFnt{\scriptsize}
\usepackage{amsfonts}
\usepackage{makecell}
\usepackage{bm}

\usepackage{multirow}
\usepackage{array}
\usepackage{makecell}
\usepackage{graphicx}
\newcolumntype{P}[1]{>{\raggedright\arraybackslash}p{#1}}

\SetCommentSty{mycommfont}

\newcommand{\fsync}{$\mathcal {F}$\textsc{sync}}
\newcommand{\ssync}{$\mathcal {S}$\textsc{sync}}

\journal{Arxiv}

\usepackage{a4wide}
\begin{document}
\usetikzlibrary{arrows}
\begin{frontmatter}
\title{Semi-Synchronous Exploration in Dynamic Graphs}

            


\author[inst1]{Ashish Saxena}

\author[inst2]{Anisur Rahaman Molla}

\author[inst1]{Kaushik Mondal}
\author[inst3]{Gokarna Sharma}

\affiliation[inst1]{organization={Indian Institute of Technology Ropar}, 
            city={Rupnagar},
            postcode={140001}, 
            state={Punjab},
            country={India}}
            
\affiliation[inst2]{organization={Indian Statistical Institute Kolkata}, 
            city={Kolkata},
            postcode={700108}, 
            state={West Bengal},
            country={India}}

\affiliation[inst3]{organization={Kent State University}, 
            city={Kent},
            postcode={44242}, 
            state={Ohio},
            country={USA}}

\begin{abstract}
We study the fundamental problem of graph exploration in dynamic graphs using mobile agents. We consider $1$-interval connected dynamic graphs, where the topology may change arbitrarily from round to round as long as the graph remains connected, and edges are assigned with the dynamic port labeling at each round. The execution follows a semi-synchronous scheduler, under which an adversary may deactivate an arbitrary subset of agents in each round. For a graph with $n$ nodes and $k$ agents, we show that exploration is impossible if the adversary can deactivate at least $ \left\lceil \frac{k}{n-2} \right\rceil - 1$ agents per round, even when agents are equipped with unbounded memory, have global communication and full visibility. This yields an upper bound, implying that exploration is solvable only when the adversary deactivates at most $\left\lceil \frac{k}{n-2} \right\rceil - 2$ agents per round. We further establish that achieving exploration at this threshold requires agents to have both $1$-hop visibility and $1$-hop communication. Finally, we present the exploration algorithm using $k$ agents when the adversary deactivates at most $ \left\lceil \frac{k}{n-2} \right\rceil - 2$ agents, assuming agents are equipped with $1$-hop visibility and global communication, and matches the adversarial deactivation bound implied by the impossibility results.
\medskip

\end{abstract}

\end{frontmatter}

\section{Introduction}
The exploration of graphs by mobile agents is a fundamental problem in distributed computing. Since the early work of Shannon~\cite{shannon1993}, the goal has been to design algorithms that enable agents to collectively visit all nodes of a network. This problem arises in autonomous systems such as mobile robots, software agents, and vehicular networks, where agents must gather information, detect faults, or disseminate data. Graph exploration has been extensively studied under a variety of assumptions (for a survey, see~\cite{Das_2019}), predominantly in static graphs where the topology remains fixed. In many modern systems, however, network topologies evolve over time (dynamic graphs), motivating the study of exploration in dynamic graphs. In such settings, maintaining connectivity is essential, as isolated nodes cannot be reached by any agent. Several models have been proposed to capture connectivity in dynamic graphs \cite{Kuhn_2010,Casteigts_2012,Michail_2014,saxena_2025}.

A graph \( G \) is \emph{port-labeled} if each edge incident to a node \( u \) is assigned a locally unique port number from the set \( \{0,1,\ldots,\deg(u)-1\} \), where $\deg(u)$ is the degree of node $u$ in $G$. An undirected edge \( \{u,v\} \) thus has two independent port numbers, one at each endpoint. The question of how port numbers behave across rounds becomes even more critical in dynamic graphs. In the literature on dynamic graphs involving mobile agents, mainly two port-labeling models have been studied. The first, we call as \textit{fixed port labeling} \cite{gotoh2018group, gotoh2019exploration, GOTOH2021, di2020distributed, bournat2017computability}, the network has a static footprint graph \( G \), and each \( \mathcal{G}_r \) is a subgraph of \( G \). The outgoing ports at node $u$ in $\mathcal{G}_r$ are the same as the outgoing ports of node $u$ in $G$. Specifically, if port \( \lambda \) at node \( u \) leads to node \( v \) in the footprint $G$, then whenever that edge corresponding port appears in any \( \mathcal{G}_r \), it always leads to node \( v \). This provides local stability in the network. The second is referred to as \textit{dynamic port labelling} \cite{saxena_2025,Ajay_dynamicdisp,Saxena_DISC}, where there is no footprint. Therefore, the degree of each node can change arbitrarily in each round, and port numbers are reassigned dynamically based on the current degree. As a result, the port $\lambda$ at node $u$ may lead to different nodes at different rounds.

Dynamic port labeling models networks such as wireless or mobile ad hoc networks, where no persistent underlying structure can be assumed and the neighborhood of a node may change arbitrarily between rounds. Since port numbers are reassigned independently in each round based on the current local topology, port $\lambda$ at node $u$ may lead to different neighbors in different rounds, providing no consistency across rounds to a traversing agent. Thus, dynamic port labeling captures a strictly more adversarial and general setting than fixed port labeling. In particular, any algorithm correct under dynamic port labeling is also correct under fixed port labeling, whereas the converse does not necessarily hold.




Our work studies the exploration problem in a $1$-interval connected $n$-node dynamic network with dynamic port labeling, where $1$-interval connected means that $\mathcal{G}_r$ is connected in every round $r$. We consider the semi-synchronous (\ssync) scheduler, in which an adversary may deactivate an arbitrary subset of agents in each round. In the next section, we present the model and the formal problem definition.

\subsection{Model and problem definition}
\noindent \textbf{Dynamic graph}: We consider a dynamic network modeled as a sequence of undirected graphs $\mathcal{G} = (V, E)$, where the node set $V$ remains fixed over time and satisfies $|V| = n$. Define $S = \big\{\{u, v\}\,|\,u, v \in V \big\}$ as the set of all possible edges, and let $\mathscr{P}(S)$ denote its power set. The function $E: \mathbb{N} \rightarrow \mathscr{P}(S)$ maps each round number $r \in \mathbb{N} \cup \{0\}$ to the set of edges $E(r)$ present at that round, yielding the snapshot graph $\mathcal{G}_r = (V, E(r))$. The dynamic graph $G$ is thus given as a sequence $\langle \mathcal{G}_0, \mathcal{G}_1, \mathcal{G}_2, \ldots \rangle$. We assume the presence of a dynamic adversary that may insert or delete any edge at the beginning of each round. For a node $v \in V$, the degree of node $v$ in $\mathcal{G}_r$ is defined as $\deg_r(v)$, which we call \emph{dynamic degree}. We define \emph{dynamic diameter} $\hat{D}:=\max_{r\geq 0} D_r$, where $D_r$ is the diameter of $\mathcal{G}_r$.

Each snapshot graph $\mathcal{G}_r$ is unweighted, undirected, and anonymous. Moreover, the graph is port-labelled: for any node $v \in \mathcal{G}_r$, the incident edges are assigned distinct local port numbers in the range $[0, \deg_r(v)-1]$. For an edge $\{u, v\}$, the port numbers at $u$ and $v$ are independently assigned and unrelated. Port labellings can differ across rounds; i.e., port numbers at a node in $\mathcal{G}_r$ may not match those in $\mathcal{G}_{r'}$ for $r \ne r'$. Nodes do not have any storage capability. In this work, the graph $\mathcal{G}(=\langle \mathcal{G}_0, \mathcal{G}_1, \mathcal{G}_2, \ldots \rangle)$ maintains the \textit{1-interval connectivity}, which says that $\mathcal{G}_r$ is connected at every round $r$.

\medskip
\noindent\textbf{Agent, cycle, hole, multinode}: We consider $k$ agents initially placed at nodes. Each agent has a unique identifier from the range $[1, n^{c}]$, for a constant $c$. An agent knows only its own ID, and is unaware of the values of $n$, $k$, and $c$. Agents are equipped with memory. An agent at node $v$ in round $r$ knows the set of ports incident to $v$ in $\mathcal{G}_r$. The algorithm proceeds in synchronous rounds. In each round $r$, every agent executes one \emph{Communicate--Compute--Move} cycle. During the \emph{communicate}, an agent exchanges information with other agents according to its communication capability. In the \emph{compute}, it computes based on the available information. Finally, in the \emph{move}, the agent either moves to a neighboring node or remains at its current node according to the computed decision. A node $v \in \mathcal{G}_r$ is called a {\em hole} in round $r$ if it contains no agent, and a {\em multinode} if it contains two or more agents.

\medskip
\noindent\textbf{Activation, fairness, move complexity}:
We consider the {\ssync} scheduler, where in each round an arbitrary subset of agents is activated. The activation pattern is controlled by an adversary with full knowledge of the agents’ algorithm and unbounded computational power, subject to the \emph{fairness} condition, i.e., every agent is activated infinitely often. An agent not activated in round $r$ is said to be \emph{inactive} in that round. While each deactivation period is finite, its duration is unbounded, and an active agent has no information about its previous activations or events during inactivity. Due to unbounded deactivation periods, we measure time complexity using \emph{move complexity} rather than round complexity. A \emph{move} is the traversal of an edge by an agent during a round, and the move complexity is the total number of such moves made by all agents until task completion.

\medskip
\noindent \textbf{Communication:}
In the $\ell_c$-hop communication model, where $\ell_c \in [0, \hat{D}]$, an agent at a node $v$ can communicate to all agents located within $\ell_c$ hops from $v$. When $\ell_c = 0$, this corresponds to \emph{face-to-face (f-2-f)} communication. When $\ell_c=\hat{D}$, the model corresponds to \emph{global communication}.

\medskip
\noindent \textbf{Visibility:} 
We consider an $\ell_v$-hop visibility model, where $\ell_v \in [0, \hat{D}]$. An agent located at a node $w \in \mathcal{G}_r$ can observe all nodes at distance at most $\ell_v$ from $w$ in round $r$. In particular, it sees the subgraph of $\mathcal{G}_r$ induced by these nodes, along with the identifiers of the agents located at them. However, an agent cannot access the internal memory of other agents nor determine whether they are active or inactive. When $\ell_v = \hat{D}$, the model corresponds to \emph{full visibility}.

\medskip
\noindent \textbf{Problem definition:} \textit{\textbf{(Exploration)} A node \( v \) is said to be \emph{visited by round \( r \)} if there exists a round \( t \in [0,r] \) such that at least one agent is located at \( v \) in round \( t \). An algorithm achieves \emph{exploration} if every node is visited at least once.}

\subsection{Technical challenges, ideas and contributions}\label{sec:tech}
The {\ssync} setting in dynamic graphs allows both agent activity and network structure to be controlled adversarially, making coordination and progress detection inherently difficult. Assume that the adversary can deactivate at most $p (\geq 1)$ agents per round. Consider a system with $(n-2)(p+1)$ agents. If there is a hole $v_h$ at round $r$, then among the remaining $n-1$ nodes, either there exists another hole, or there are at least two nodes that contain at most $p$ agents. If the adversary may choose $\mathcal{G}_r$ arbitrarily in each round, it can always attach $v_h$ as a pendant to such a node. If this node is a hole, it is at least two hops away from the active agents; if it contains at most $p$ agents, the adversary can deactivate all agents at that node, ensuring that $v_h$ remains unreachable by active agents. Thus, in every round, the adversary can maintain a node which remains unvisited. Based on this, we have the following result.

\medskip
\noindent\fbox{
\parbox{0.98\textwidth}{
\textbf{Result 1.} (refer to Theorem \ref{thm:imp} and Observation \ref{obs})
If the adversary can deactivate at most $p$ agents per round, exploration is impossible with
$k \le (n-2)(p+1)$ agents, even with full visibility, unbounded node storage, unique identifier to the nodes, global communication, unbounded memory, and complete parameter knowledge. Equivalently, for a given $k$, exploration is impossible if
$p \ge \left\lceil \frac{k}{n-2} \right\rceil - 1$.
Hence, exploration with $k$ agents is solvable only when
$p \le \left\lceil \frac{k}{n-2} \right\rceil - 2$.
}}

\medskip
Due to Result 1, if the adversary can deactivate at most $p$ agents per round, at least $(p+1)(n-2)+1$ agents are necessary. Consider a special configuration in which exactly $(n-2)$ nodes each contain $(p+1)$ agents and exactly one node contains a single agent; we denote this \underline{configuration by $\mathcal{C}^*$}. Suppose that at some round $r$ the agents are in $\mathcal{C}^*$ and are equipped with $1$-hop visibility. There is exactly one hole, $v_h$. If at some round $r_1 \ge r$ the hole $v_h$ is adjacent to a node containing $p+1$ agents, then since the adversary can deactivate at most $p$ agents, at least one agent at that node is active and can observe the hole using $1$-hop visibility and move to $v_h$. If $v_h$ is never adjacent to a node containing $p+1$ agents, then in every round it is adjacent to the node containing the single agent. Due to the fairness, the agent is eventually activated in some round and can move to $v_h$ using its $1$-hop visibility. Hence, in all cases, the hole is eventually visited. If an algorithm does not reach $\mathcal{C}^*$, the adversary can always select a node containing at most $p$ agents and deactivate all agents at that node. This shows that $\mathcal{C}^*$ is necessary to achieve the exploration using $(p+1)(n-2)+1$ agents. Consider agents equipped with $1$-hop visibility and f-2-f communication. If the initial configuration is not $\mathcal{C}^*$, the adversary can prevent the system from ever reaching $\mathcal{C}^*$. The idea is as follows. The adversary maintains a path $P = v_1 \sim v_2 \sim \cdots \sim v_n$ such that $v_n$ is unvisited by round $r-1$ and $\alpha(v_i) \ge \alpha(v_{i+1})$ for all $i \in [1,n-2]$, where $\alpha(v_i)$ is the number of agents at $v_i$ at $r$. Knowing $P$, if the adversary pre-computes that agents would reach $\mathcal{C}^*$ at the end of round $r$, it modifies the graph by deleting and adding a single edge while preserving $1$-interval connectivity. Since agents have only $1$-hop visibility, this modification affects the decisions of only a constant number of nodes. A careful analysis shows that agents fail to reach the $\mathcal{C}^*$ configuration. We have the following result.

\medskip
\noindent\fbox{
\parbox{0.98\textwidth}{
\textbf{Result 2.} (refer to Theorem \ref{thm:imp-1-hop})
Exploration is impossible to solve using $(n-2)(p+1)+1$ agents in 1-interval connected graphs when agents have $1$-hop visibility and f-2-f communication, even if nodes have unbounded storage, unique identifiers, agents have unbounded memory and complete parameter knowledge.
}}

\medskip

Consider agents equipped with global communication and $0$-hop visibility. Suppose that at some round $r$ there is a hole $v_h$. With $(p+1)(n-2)+1$ agents on the remaining nodes, either the agents are in configuration $\mathcal{C}^*$ or they are not. If the agents are not in configuration $\mathcal{C}^*$, then regardless of how the agents are distributed over the remaining $n-1$ nodes, there exist at least two nodes, say $v_1$ and $v_2$, each containing at most $p$ agents. In this case, the adversary can select one such node, say $v_1$, attach $v_h$ as a pendant to $v_1$, and deactivate all agents at $v_1$. Consequently, no active agent can reach $v_h$. If the agents are in configuration $\mathcal{C}^*$, the adversary attaches $v_h$ as a pendant to the node $v$ that contains exactly one agent. Even if all agents are active, the adversary, knowing the agents’ algorithm, can precompute the movement of the agent at node $v$ and modify the port labeling so that the port chosen by the agent does not lead to $v_h$. Since agents are equipped with $0$-hop visibility and global communication, they cannot detect such a change in ports, and the agent at node $v$ follows the same computation. In this way, the adversary can prevent exploration even when all agents are active. Based on this, we have the following result.

\medskip
\noindent\fbox{
\parbox{0.99\textwidth}{
\textbf{Result 3.}
(refer to Theorem \ref{thm:imp-global}) Exploration is impossible to solve using $(n-2)(p+1)+1$ agents in 1-interval connected graphs when agents have $0$-hop visibility and global communication, even if nodes have unbounded storage, unique identifiers, agents have unbounded memory and complete parameter knowledge.
}}

\medskip

Since {\fsync} is stronger than {\ssync}, the lower bound on $k$ under {\fsync} also applies to {\ssync}. In~\cite{saxena_2025}, the model is the same as ours, and it is shown that exploration is impossible with $n-2$ agents under {\fsync}, even with full visibility, global communication, unique node identifiers, unbounded memory, unbounded node storage, and complete knowledge of parameters. Hence, $k \ge n-1$ under {\ssync}. Using Results~1, 2, 3 and $k \ge n-1$, we obtain the following.

\medskip
\noindent\fbox{
\parbox{0.99\textwidth}{
\textbf{Result 4.}
(refer to Observation \ref{obs:necessary}) To solve exploration with $(p+1)(n-2)+1$ agents, 1-hop visibility and 1-hop communication are necessary. Consequently, if exploration is solvable with $k$ agents under an adversary deactivating at most $p$ agents per round, then
$k \ge (p+1)(n-2)+1$, which implies
$p \le \left\lfloor \frac{k-1}{n-2} \right\rfloor - 1$. Note that $ \left\lceil \frac{k}{n-2} \right\rceil -1= \left\lfloor \frac{k-1}{n-2} \right\rfloor$ for $k\geq n-1$ and $n\geq 3$ as equality $\left\lfloor \frac{A}{B} \right\rfloor = \left\lceil \frac{A+1}{B} \right\rceil - 1$ is valid for natural numbers $A,B$. Thus, the bound matches Result~1.
}}

\medskip

Further, we present an algorithm for exploration when agents are equipped with $1$-hop visibility and global communication, and when $k$ agents are present while the adversary can deactivate at most $p = \left\lfloor \frac{k-1}{n-2} \right\rfloor - 1$ agents per round. Recall the \textit{pipeline strategy}: if agents have complete knowledge of the snapshot $\mathcal{G}_r$, a multinode $w$ can push an agent along a shortest path to a hole $v$ by simultaneously moving one agent along each edge of the path, thereby filling the hole without creating new ones. This strategy has been used under the {\fsync} scheduler~\cite{das_2020, Ajay_dynamicdisp, saxena_2025, Saxena_DISC}; however, it becomes challenging under the {\ssync} scheduler. In particular, if all agents at some intermediate node are inactive in a round, progress toward the hole is blocked. To overcome this difficulty, our algorithm estimates $p$ progressively. Agents are initially unaware of the parameters $k$, $n$, and $p$. Using $1$-hop visibility, an agent can observe the number of agents at its neighboring nodes, although it cannot distinguish active agents from inactive ones. Through global communication, agents broadcast their $1$-hop information. If an agent observes that a neighboring node of some node containing an active agent has $c$ agents but receives no message from any of them in a given round, it can infer that all $c$ agents are inactive. Whenever $c$ exceeds the agent’s current estimate of $p$, the estimate is updated. Using this estimate, agents repeatedly apply the pipeline strategy to fill either a hole or a node containing only inactive agents. Once agents observe that every node contains at most $p+1$ agents, they fix a source node. Whenever at least one agent at the source is active, it initiates a pipeline toward a hole or a node with no active agents. We show that this process eventually leads the system either to configuration $\mathcal{C}^*$ or to a configuration in which every node contains at least one agent. Moreover, the algorithm can distinguish between these two situations. Based on this, we have the following algorithmic result.

\medskip
\noindent\fbox{
\parbox{0.98\textwidth}{
\textbf{Result 5.}
(refer to Theorem \ref{thm:final}) If $p \le  \left\lfloor \frac{k-1}{n-2} \right\rfloor - 1$ and agents have $1$-hop visibility and global communication, we present an algorithm that solves exploration under the {\ssync} scheduler. The move complexity is $O(k\hat{D})$, and each agent uses $O(\max\{\log n,\log p\})$ memory. The algorithm requires no prior knowledge of the parameters.
}}

\medskip
 In the next section, we compare our results with the existing literature on dynamic graphs.

\subsection{Related work}\label{sec:rel}
 Under fixed port labeling, the {\ssync} scheduler is introduced in \cite{di2020distributed}. Under the {\ssync} scheduler, an agent that attempts to move through a port whose corresponding edge is absent in a given round becomes inactive and remains so until reactivated in a later round. This behavior gives rise to the notion of an agent \emph{sleeping} at a port, which leads to several variants of the model: (i) \emph{no simultaneity} (NS), where a sleeping agent cannot move and has no guarantee of activation when the edge reappears; (ii) \emph{passive transport} (PT), where a sleeping agent is automatically transported if the edge is present in the same round; and (iii) \emph{eventual transport} (ET), where a sleeping agent cannot move but is guaranteed to be activated in a round in which the edge is present. In the dynamic port-labeling model considered in this work, the edge corresponding to each port is always present. Consequently, the notion of sleeping agents and sleep-based transport models is not applicable. We therefore adopt the {\ssync} model as used for static graphs, treating each snapshot $\mathcal{G}_r$ as a port-labeled static graph. In this setting, in each round $r$, an arbitrary subset of agents is active. To the best of our knowledge, there is no prior work that considers the {\ssync} scheduler in the dynamic port-labeling model.

Graph exploration has been extensively studied in the literature~\cite{flocchini2012searching, flocchini2013exploration, ilcinkas2011power, ilcinkas2018exploration, bournat2016self, bournat2017computability, di2020distributed, gotoh2018group, GOTOH2021, Saxena_DISC, saxena_2025}. Exploration in general dynamic graphs has been studied under various connectivity models, including 1-interval connected graphs~\cite{saxena_2025}, $\ell$-bounded 1-interval connected graphs and temporally connected graphs~\cite{GOTOH2021}, and connectivity time dynamic graphs~\cite{Saxena_DISC}. Among these, ET-based {\ssync} is considered in~\cite{GOTOH2021}. Since our focus is on 1-interval connected graphs and the {\ssync} scheduler, we discuss in detail the results of \cite{saxena_2025,GOTOH2021}. The dynamic graph model of~\cite{saxena_2025} is identical to ours except that agents operate under the {\fsync} scheduler. They show that $n-2$ agents are insufficient for exploration and provide an algorithm that succeeds with $n-1$ agents under $1$-hop visibility and global communication. In~\cite{GOTOH2021}, exploration is studied under the ET-based {\ssync} scheduler with fixed port labeling, assuming anonymous agents and node storage. In $\ell$-bounded $1$-interval connected graphs, they provide a tight analysis by proving that $2\ell$ agents are insufficient while $2\ell+1$ agents suffice, where $\ell$ is the maximum number of edge deletions per round. Since $\ell \le m-n+1$, $O(n^2)$ agents are always enough for exploration. In contrast, our results apply to any number of agents $k$ and provide a tight threshold on the adversary’s deactivation power as a function of $k$ and $n$. Specifically, given any $O(n^c)$ agents for $c\geq 1$, there exists a value of $p$ such that exploration becomes impossible beyond that value. Moreover, the approach in~\cite{GOTOH2021} relies on fixed port labeling and node storage: one agent waits for a missing edge while the others execute a rotor-router strategy. Because at most $\ell$ edges can be deleted, at most $2\ell$ agents become stuck, and one extra agent guarantees exploration due to the ET model. Such a strategy does not apply in our setting due to the dynamic port labeling and the lack of node storage.

\section{Impossibility results}\label{sec:imp}
In this section, we present the impossibility results. The high-level approach is outlined in Section~\ref{sec:tech}.

\begin{theorem}\label{thm:imp}
    ($n\geq 3, p\geq 1$) If the adversary deactivates at most $p$ agents per round, exploration is impossible with $(n-2)(p+1)$ agents in 1-interval connected graphs, even with full visibility, unbounded node storage, unique identifier to the nodes, global communication, unbounded memory, and complete parameter knowledge.
\end{theorem}
\begin{proof}
 Let \(v\) be a hole initially. If no hole existed initially, exploration would already be complete. We show that \(v\) remains a hole forever. For every round \(r\ge 0\), we construct a dynamic graph \(\mathcal{G}_r\) from $G$ so that all active agents are always at a distance of at least two from \(v\). The construction of $\mathcal{G}_r$ based on the number of holes is as follows. 

\medskip
\noindent\textbf{Case 1 (At least two holes at the beginning of round \(r\)):} Let $w_1$, $w_2$, \ldots, $w_n$ be nodes at the beginning of round $r$. Assume that \(w_{n-1}\) and \(w_n\) are holes. Without loss of generality, let $w_n=v$. The adversary forms a clique from nodes $w_1$, $w_2$, \ldots, $w_{n-1}$, and attach node $w_n$ as a pendant node to node $w_{n-1}$. In this case, the adversary does not deactivate any agent. In this case, all active agents are at a distance of at least two from \(w_n\).

\medskip
\noindent\textbf{Case 2 (Exactly one hole at the beginning of round \(r\)):} Let $w_1$, $w_2$, \ldots, $w_n$ be nodes at the beginning of round $r$. Assume that \(w_{n-1}\) and \(w_n\) are holes. Without loss of generality, let $w_n=v$. In this case, all agents are at $n-1$ nodes. Since $(p+1)(n-2)$ agents are present, there are at least two nodes which have $\leq p$ agents. Without loss of generality, let nodes $w_{n-2}$ and $w_{n-1}$ have at most $p$ agents. For $r\geq 1$, let $w \in \{w_{n-2}, w_{n}\}$ such that every agent was active at round $r-1$. If $r=0$, then let $w=w_{n-2}$.
The adversary forms a clique from nodes $w_1$, $w_2$, \ldots, $w_{n-1}$, and attach node $w_n$ as a pendant node to node $w_{n-2}$. In this case, the deactivates all agents at node $w_{n-2}$. In this case, all active agents are at a distance of at least two from \(w_n\).

\medskip
We now show by induction that \(v\) remains a hole for all \(r\ge 0\).

\smallskip
\noindent\emph{Base case \(r=0\).}  
In both cases above, no active agent is adjacent to \(v\). Thus \(v\) is still empty at the end of round \(0\).

\smallskip
\noindent\emph{Induction step.}  
Assume \(v\) is a hole at the end of round \(r\). At the start of round \(r+1\), if there is another hole, we apply Case~1; otherwise, we apply Case~2. In both constructions, every active agent remains at a distance of at least two from \(v\), so \(v\) remains a hole at the end of round \(r+1\).

\medskip
\noindent\textbf{Fairness.}
Whenever the adversary deactivates agents at some node, these agents are reactivated in the following round. Thus, every agent is active infinitely often, and the schedule is fair.

\medskip
Since the proof does not rely on any limitations, it remains valid even if nodes have unbounded storage and unique identifiers, and agents have unbounded memory, full visibility, global communication and complete parameter knowledge. This completes the proof.
\end{proof}

\begin{observation}\label{obs}
If the adversary deactivates at most $p$ agents per round, exploration is impossible to solve with
$k \le (n-2)(p+1)$ agents. Equivalently, for a given $k$, exploration is impossible if
$p \ge \left\lceil \frac{k}{n-2} \right\rceil - 1$.
Hence, exploration with $k$ agents is solvable only when
$p \le \left\lceil \frac{k}{n-2} \right\rceil - 2$.    
\end{observation}

By Theorem~\ref{thm:imp}, at least $(p+1)(n-2)+1$ agents are required when the adversary can deactivate at most $p$ agents per round. We now study the necessary assumptions with $(p+1)(n-2)+1$ agents. The proof of the following theorem is non-trivial, lengthy, and relies on extensive case analysis.

\begin{theorem}\label{thm:imp-1-hop}
($n\geq 9, p\geq 1$) If the adversary deactivates at most $p$ agents per round, exploration is impossible with $(n-2)(p+1)+1$ agents in 1-interval connected graphs with $1$-hop visibility and f-2-f communication, even if nodes have unique identifiers and unbounded storage, and agents possess unbounded memory and complete parameter knowledge. 
\end{theorem}
\begin{proof}
Let $G$ be a clique of size $n$. Consider any initial configuration with at least one hole, and agents are not in a $\mathcal{C}^*$ configuration. Label the nodes $v_1,\dots,v_n$ such that $\alpha(v_i)\ge \alpha(v_{i+1})$ for all $i\in[1,n-1]$, where $\alpha(v_j)$ denotes the initial number of agents at $v_j$. As discussed in Section \ref{sec:tech}, if agents are in $\mathcal{C}^*$ configuration, agents can explore $G$ using 1-hop visibility within finite but unbounded time. Therefore, we construct $\mathcal{G}_r$ so that the agents never reach configuration $\mathcal{C}^*$ at the end of round $r$, which leads exploration to fail. We use two notations for nodes. Let $H_r$ be a graph at round $r$. For any node $v \in H_r$, let $S_{H_r}(v)$ and $E_{H_r}(v)$ denote the number of agents at the beginning and end of round $r$ at $v$, respectively.
\begin{figure}
    \centering
    \includegraphics[width=0.75\linewidth]{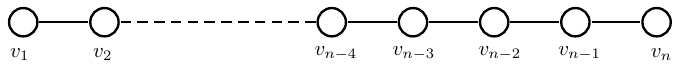}
    \caption{The construction of $\mathcal{C}_0$.}
    \label{fig:1-hop_1}
\end{figure}

\vspace{0.2cm}
\noindent \textbf{Graph $\mathcal{C}_r$ at round $r\geq 0$:}
    \begin{itemize}
         \item $\bm{r=0}:$  Initial configuration is not $\mathcal{C}^*$, and there is at least one hole. Without loss of generality, let $v_n$ be a hole at the beginning of round $r=0$. Consider a path $P_0:=v_1\sim v_2\sim v_3\sim \ldots v_{n-1}\sim v_n$.  In this case, $\mathcal{C}_0=P_0$ (refer to Fig. \ref{fig:1-hop_1}). The adversary deactivates the agents as follows.
         \begin{itemize}
             \item If there are at least one agent at node $v_{n-1}$, the adversary deactivates all agents at node $v_{n-1}$, and all other agents are active. 
             \item If there is no agent at node $v_{n-1}$, then it deactivates one agent at node $v_{n-2}$.
         \end{itemize}
         
        \item $\bm{r\geq1}:$ At the end of round $r-1$, let $w_1$, $w_2$, \ldots,$w_{n-1}$, $w_n(=v_n)$ be nodes. In $\mathcal{G}_{r-1}$, agents do not achieve $\mathcal{C}^*$. Note that in $\mathcal{G}_{r-1}$, the value of $deg_r(w_n)$ is 1 (we observe this fact after completing the construction of $\mathcal{G}_r$). We denote the neighbour of $w_n$ by $w_{n-1}$ in $\mathcal{G}_{r-1}$. There are two possible cases: 
\begin{figure}
    \centering
    \includegraphics[width=0.75\linewidth]{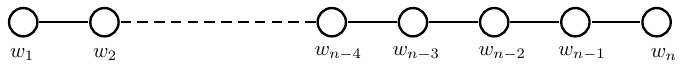}
    \caption{The construction of $\mathcal{C}_r$ when $E_{\mathcal{G}_{r-1}}(w_{n-1})=0$.}
    \label{fig:1-hop_2}
\end{figure}

        \vspace{0.2cm}
        \noindent \textbf{Case 1} $(\bm{E_{\mathcal{G}_{r-1}}(w_{n-1})=0)}$: Without loss of generality, let $w_1$, $w_2$, \ldots,$w_{n-2}$ be nodes such that $E_{\mathcal{G}_{r-1}}(w_i)\geq E_{\mathcal{G}_r}(w_{i+1})$ for every $i\in [1,n-3]$. Consider a path, $P_r:=w_1\sim w_2\sim w_3\sim \ldots w_{n-1}\sim w_n$. In this case, $\mathcal{C}_r=P_r$ (refer to Fig. \ref{fig:1-hop_2}). In this case, the adversary deactivates agents as follows.
        \begin{itemize}
            \item If $E_{\mathcal{G}_r}(w_{n-2})\leq p$, then it does not deactivate any agent.
            \item If $E_{\mathcal{G}_r}(w_{n-2})\geq p+1$, then it deactivates one agent at node $w_{n-2}$ which was not deactivated in round $r-1$. Such agent always exists as $E_{\mathcal{G}_r}(w_{n-2})\geq p+1$.
        \end{itemize}

\begin{figure}
    \centering
    \includegraphics[width=0.75\linewidth]{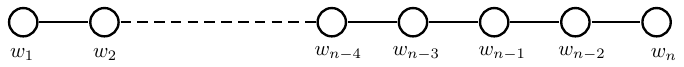}
    \caption{The construction of $\mathcal{C}_r$ when $E_{\mathcal{G}_{r-1}}(w_{n-1})\neq 0$.}
    \label{fig:1-hop_3}
\end{figure}
         \vspace{0.2cm}
        \noindent \textbf{Case 2} $\bm{(E_{\mathcal{G}_{r-1}}(w_{n-1})\neq 0)}$: Since agents do not achieve $\mathcal{C}^*$ at the end of round $r-1$, there is at least one node from set $\{w_1,w_2, \ldots, w_{n-2}\}$ which has $\leq p$ agents. Without loss of generality, let $E_{\mathcal{G}_{r-1}}(w_{n-2})\leq p$. Consider $w_1$, $w_2$, \ldots,$w_{n-3}$, $w_{n-1}$ be nodes such that $E_{\mathcal{G}_{r-1}}(w_i)\geq E_{\mathcal{G}_r}(w_{i+1})$ for every $i\in [1,n-4]$ and $E_{\mathcal{G}_{r-1}}(w_{n-3})\geq E_{\mathcal{G}_{r-1}}(w_{n-1})$. Consider a path, $P_r:=w_1\sim w_2\sim w_3\sim \ldots w_{n-3}\sim w_{n-1}\sim w_{n-2}\sim w_n$. In this case, $\mathcal{C}_r=P_r$ (refer to Fig. \ref{fig:1-hop_3}). The adversary deactivates all the agents at node $w_{n-1}$.

        \end{itemize}
        As per the aforementioned construction of $\mathcal{C}_r$ for round $r\ge 0$, the graph is a path $P = u_1 \sim u_2 \sim \ldots \sim u_{n-2} \sim u_{n-1} \sim u_n(=v_n)$. The values $S_{\mathcal{C}_r}(u_i)$ do not increase as we move from $u_i$ to $u_{i+1}$ for every $i\in[1,n-3]$, and the agents are not in configuration $\mathcal{C}^*$. 

        Since the adversary knows the agents' algorithm and the structure of $\mathcal{C}_r$, it can pre-compute the agents' movement and determine whether they would reach configuration $\mathcal{C}^*$ when $\mathcal{G}_r=\mathcal{C}_r$ at the end of round $r$. If they do not reach $\mathcal{C}^*$, the adversary simply forms $\mathcal{C}_r$ at the beginning of round $r$. If they would reach $\mathcal{C}^*$, then the adversary instead forms $\mathcal{C}_r'$ at the beginning of round $r$. We show that whenever the agents would reach $\mathcal{C}^*$ under $\mathcal{C}_r$, they fail to reach $\mathcal{C}^*$ under $\mathcal{C}_r'$ at the end of the same round. The construction of $\mathcal{C}_r'$ depends on the value of $S_{\mathcal{C}_r}(u_{\,n-1})$. The reason why the agents fail to reach configuration $\mathcal{C}^*$ when $\mathcal{G}_r=\mathcal{C}_r'$ is described below. 

        Let $S_{\mathcal{C}_r}(u_{n-1})=Y$, $S_{\mathcal{C}_r}(u_{n-2})=X_1$, $S_{\mathcal{C}_r}(u_{n-3})=X_2$, and $S_{\mathcal{C}_r}(u_{n-4})=X_3$.  
Based on the pre-computation of agent movements in $\mathcal{C}_r$ is as follows (refer to Fig. \ref{fig:1-hop_4}):
\begin{figure}
    \centering
    \includegraphics[width=0.75\linewidth]{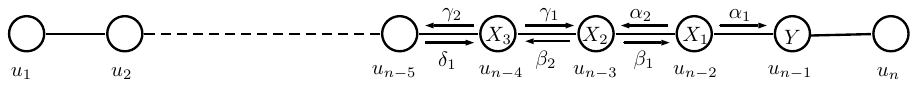}
    \caption{The pre-computation of agents in $\mathcal{C}_r$.}
    \label{fig:1-hop_4}
\end{figure}
\begin{itemize}
    \item From node $u_{n-2}$: $\alpha_1$ agent(s) move to $u_{n-1}$, and $\alpha_2$ agent(s) move to $u_{n-3}$.
    \item From node $u_{n-3}$: $\beta_1$ agent(s) move to $u_{n-2}$, and $\beta_2$ agent(s) move to $u_{n-4}$.
    \item From node $u_{n-4}$: $\gamma_1$ agent(s) move to $u_{n-3}$, and $\gamma_2$ agent(s) move to $u_{n-5}$.
    \item From node $u_{n-5}$: $\delta_2$ agent(s) move to $u_{n-4}$.
\end{itemize}

\vspace{0.1cm}
\begin{figure}
    \centering
    \includegraphics[width=0.75\linewidth]{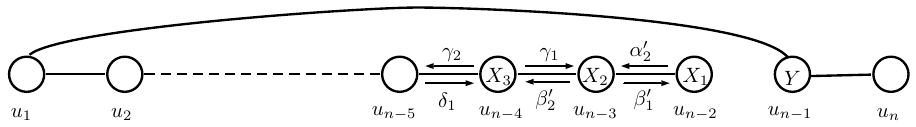}
    \caption{The construction of $\mathcal{C}_r'$ when $2 \le Y \le p$.}
    \label{fig:1-hop_5}
\end{figure}
\noindent \underline{Case 1 $(2 \le Y \le p)$}:  
In this case we have $\alpha_1 = p+1-Y$, because all agents at node $u_{n-1}$ are deactivated in $\mathcal{C}_r$. Thus, the only way to reach configuration $\mathcal{C}^*$ at the end of round $r$ is to increase the number of agents at $u_{n-1}$ to exactly $p+1$. Due to $2\leq Y\leq p$, the following inequality holds.
\begin{align}
    1\leq \alpha_1\leq p-1\label{eq:11}
\end{align}
As per pre-computation, the following equations should hold at the end of round $r$ in $\mathcal{C}_r$.
\begin{align}
    X_1-(\alpha_1+\alpha_2)+\beta_1&= 1 \text{ or } p+1\\
    X_2-(\beta_1+\beta_2)+(\alpha_2+\gamma_1)&= 1 \text{ or } p+1
\end{align}
The adversary constructs $\mathcal{C}_r'$ as follows. In path $P = u_1 \sim u_2 \sim \ldots \sim u_{n-2} \sim u_{n-1} \sim u_n(=v_n)$, the adversary removes edge between node $u_{n-2}$ and $u_{n-1}$, and add an edge between node $u_1$ and $u_{n-1}$. In configuration $\mathcal{C}_r'$, the adversary deactivates exactly the same agents that were deactivated in $\mathcal{C}_r$. This new graph is $\mathcal{C}_r'$, and constructed as $\mathcal{G}_r$ at the beginning of round $r$. In this graph, the 1-hop view of agents at node $u_1$, $u_2$, $u_{n-3}$, $u_{n-2}$, $u_{n-1}$ is changed, and the remaining agents at nodes $u_i$, $i\in [3, n-4]$ have the same 1-hop view as in $\mathcal{C}_r$ as we consider $n\geq 9$. Since agents are equipped with 1-hop visibility and f-2-f communication, the movement of agents at nodes $u_1$, $u_2$, $u_{n-3}$, $u_{n-2}$, $u_{n-1}$ in $\mathcal{C}_r'$ may not be the same as $\mathcal{C}_r$, and other remaining agent's movement in $\mathcal{C}_r'$ remains the same as in $\mathcal{C}_r$. Therefore, let $\alpha_2'$ agent(s) move from node $u_{n-2}$ to $u_{n-3}$, and from node $u_{n-3}$, let $\beta_1'$ agent(s) move to $u_{n-2}$ and $\beta_2'$ agent(s) move to $u_{n-4}$. We can see this in Fig. \ref{fig:1-hop_5}.

\begin{itemize}
    \item \underline{Case 1.1}: As per pre-computation in $\mathcal{C}_r$, the following equations hold.
    \begin{align}
    X_1-(\alpha_1+\alpha_2)+\beta_1&= 1 \label{eq:1}\\
    X_2-(\beta_1+\beta_2)+(\alpha_2+\gamma_1)&= p+1\label{eq:2}
\end{align}
After taking the sum of Eq. \ref{eq:1} and Eq. \ref{eq:2}, we get the following equation. 
\begin{align}
    X_1+X_2-\beta_2-\alpha_1+\gamma_1=p+2 \label{eq:3}
\end{align}
The agent's movement in $\mathcal{C}_r'$ leads to the following equations.
    \begin{align}
    X_1-\alpha_2'+\beta_1'&= 1 \text{ or } p+1\\
    X_2-(\beta_1'+\beta_2')+(\alpha_2'+\gamma_1)&= 1 \text{ or }p+1
\end{align}
We show in each case that movement of agents in $\mathcal{C}_r'$ does not lead to $\mathcal{C}^*$. 
\begin{itemize}
    \item \underline{Case 1.1.1}: Consider the following equations hold. 
        \begin{align}
            X_1-\alpha_2'+\beta_1'&= p+1 \label{eq:4}\\
            X_2-(\beta_1'+\beta_2')+(\alpha_2'+\gamma_1)&= 1\label{eq:5}
        \end{align}
        Since in $\mathcal{C}_r$ we have $E_{\mathcal{C}_r}(u_{n-4})=p+1$, and in $\mathcal{C}_r'$ we have $E_{\mathcal{C}_r'}(u_{n-2})=1$ and $E_{\mathcal{C}_r'}(u_{n-3})=p+1$, the only way to reach $\mathcal{C}^*$ in $\mathcal{C}_r'$ that $E_{\mathcal{C}_r'}(u_{n-4})$ also becomes $p+1$ is to keep $\beta_2=\beta_2'$ as $n\geq 8$. Hence, after taking the sum of Eq. \ref{eq:4} and Eq. \ref{eq:5}, we get the following equation. 
        \begin{align}
            X_1+X_2-\beta_2+\gamma_1=p+2 \label{eq:6}
        \end{align}
        From Eq. \ref{eq:3} and Eq. \ref{eq:6}, we have the following.
        \begin{align}
            X_1+X_2-\beta_2-\alpha_1+\gamma_1=X_1+X_2-\beta_2+\gamma_1\implies \alpha_1=0 \label{eq:7}
        \end{align}
        Eq. \ref{eq:11} and Eq. \ref{eq:7} lead to the contradiction. 
    \item \underline{Case 1.1.2}: Consider the following equations hold. 
        \begin{align}
            X_1-\alpha_2'+\beta_1'&= p+1 \label{eq:8}\\
            X_2-(\beta_1'+\beta_2')+(\alpha_2'+\gamma_1)&= 1\label{eq:9}
        \end{align}
        Since in $\mathcal{C}_r$ we have $E_{\mathcal{C}_r}(u_{n-4})=p+1$, and in $\mathcal{C}_r'$ we have $E_{\mathcal{C}_r'}(u_{n-2})=1$ and $E_{\mathcal{C}_r'}(u_{n-3})=p+1$, the only way to reach $\mathcal{C}^*$ in $\mathcal{C}_r'$ that $E_{\mathcal{C}_r'}(u_{n-4})$ also becomes $p+1$ is to keep $\beta_2=\beta_2'$ as $n\geq 8$. Hence, after taking the sum of Eq. \ref{eq:8} and Eq. \ref{eq:9}, we get the following equation. 
        \begin{align}
            X_1+X_2-\beta_2+\gamma_1=p+2 \label{eq:10}
        \end{align}
        From Eq. \ref{eq:3} and Eq. \ref{eq:10}, we have the following.
        \begin{align}
            X_1+X_2-\beta_2-\alpha_1+\gamma_1=X_1+X_2-\beta_2+\gamma_1\implies \alpha_1=0 \label{eq:12}
        \end{align}
        Eq. \ref{eq:11} and Eq. \ref{eq:12} lead to the contradiction. 
        \item \underline{Case 1.1.3}: Consider the following equations hold. 
        \begin{align}
            X_1-\alpha_2'+\beta_1'&= p+1 \label{eq:13}\\
            X_2-(\beta_1'+\beta_2')+(\alpha_2'+\gamma_1)&= p+1\label{eq:14}
        \end{align}
        Since in $\mathcal{C}_r$ we have $E_{\mathcal{C}_r}(u_{n-4})=p+1$, and in $\mathcal{C}_r'$ we have $E_{\mathcal{C}_r'}(u_{n-2})=p+1$ and $E_{\mathcal{C}_r'}(u_{n-3})=p+1$, the value of $E_{\mathcal{C}_r'}(u_{n-4})$ is either $p+1$ or 1. The value of $E_{\mathcal{C}_r'}(u_{n-4})=1$ if $\beta_2'=\beta_2-p$, and the value of $E_{\mathcal{C}_r'}(u_{n-4})=p+1$ if $\beta_2'=\beta_2$ as the view of node $u_{n-5}$ and $u_{n-4}$ are the same in $\mathcal{C}_r'$. Hence, after taking the sum of Eq. \ref{eq:13} and Eq. \ref{eq:14}, we get the following equation. 
        \begin{align}
            X_1+X_2-\beta_2'+\gamma_1=2p+2 \label{eq:15}
        \end{align}
        From Eq. \ref{eq:3} and Eq. \ref{eq:15}, we have the following.
        \begin{align*}
            p+X_1+X_2-\beta_2-\alpha_1+\gamma_1=X_1+X_2-\beta_2'+\gamma_1
            \implies \alpha_1=p-\beta_2+\beta_2'  
        \end{align*}
        \begin{align}
            \implies \alpha_1=0 \text{ if } \beta_2'=\beta_2-p \text{ and } \alpha_1=p \text{ if } \beta_2'=\beta_2 \label{eq:16}
        \end{align}
        Eq. \ref{eq:11} and Eq. \ref{eq:16} lead to the contradiction. 

        \item \underline{Case 1.1.4}: Consider the following equations hold. 
        \begin{align}
            X_1-\alpha_2'+\beta_1'&= 1 \\
            X_2-(\beta_1'+\beta_2')+(\alpha_2'+\gamma_1)&= 1
        \end{align}
        This case shows that agents are not in $\mathcal{C}^*$ form at the end of round $ r$. As in $\mathcal{C}^*$, there is exactly one node with one agent.
\end{itemize}
\item \underline{Case 1.2}: As per pre-computation in $\mathcal{C}_r$, the following equations hold.
    \begin{align*}
    X_1-(\alpha_1+\alpha_2)+\beta_1&= p+1\\
    X_2-(\beta_1+\beta_2)+(\alpha_2+\gamma_1)&= 1
\end{align*}
In this case, the analysis is analogous to Case 1.1.
\item \underline{Case 1.3}: As per pre-computation in $\mathcal{C}_r$, the following equations hold.
    \begin{align}
    X_1-(\alpha_1+\alpha_2)+\beta_1&= p+1\label{eq:17}\\
    X_2-(\beta_1+\beta_2)+(\alpha_2+\gamma_1)&= p+1\label{eq:18}
\end{align}
After taking the sum of Eq. \ref{eq:17} and Eq. \ref{eq:18}, we get the following equation. 
\begin{align}
    X_1+X_2-\beta_2-\alpha_1+\gamma_1=2p+2 \label{eq:19}
\end{align}
In this sub-case, the following equality holds at node $u_{n-4}$ in $\mathcal{C}_r$.
\begin{align}
    X_3-(\gamma_1+\gamma_2)+\delta_2+\beta_2&=1 \text{ or } p+1
\end{align}
The agent's movement in $\mathcal{C}_r'$ leads to the following equations.
    \begin{align}
    X_1-\alpha_2'+\beta_1'&= 1 \text{ or } p+1\\
    X_2-(\beta_1'+\beta_2')+(\alpha_2'+\gamma_1)&= 1 \text{ or }p+1
\end{align}

\begin{itemize}
    \item \underline{Case 1.3.1}: Consider the following holds.
\begin{align}
    X_3-(\gamma_1+\gamma_2)+\delta_2+\beta_2&=1 \label{eq:20}
\end{align}
In this case, the value of $\beta_2'$ is either $\beta_2$ or $\beta_2+p$ as $n\geq 8$; otherwise the number of agents at $u_{n-4}$ would be neither $p+1$ nor $1$ in $\mathcal{C}_r'$. 
\begin{itemize}
    \item \underline{Case 1.3.1.1}: Consider the following equations hold. 
        \begin{align}
            X_1-\alpha_2'+\beta_1'&= 1 \label{eq:21}\\
            X_2-(\beta_1'+\beta_2')+(\alpha_2'+\gamma_1)&= p+1\label{eq:22}
        \end{align}
        If the value of $\beta_2'=\beta_2$, there will be two nodes (node $u_{n-4}$ and node $u_{n-2}$) in $\mathcal{C}_r'$ which have only one agent, which is not a $\mathcal{C}^*$. If the value of $\beta_2'=\beta_2+p$, then due to the following reasoning, we have a contradiction. After taking the sum of Eq. \ref{eq:21} and Eq. \ref{eq:22}, we get the following equation. 
        \begin{align}
             X_1+X_2-\beta_2'+\gamma_1=p+2 \implies X_1+X_2-\beta_2+\gamma_1=2p+2 \text{ as } \beta_2'=\beta_2+p\label{eq:23}
        \end{align}
        Due to Eq. \ref{eq:19} and Eq. \ref{eq:23}, we have the following equation.
        \begin{align}
            X_1+X_2-\beta_2-\alpha_1+\gamma_1= X_1+X_2-\beta_2+\gamma_1\implies \alpha_1=0\label{eq:231}
        \end{align}
        Due to Eq. \ref{eq:11} and Eq. \ref{eq:231}, we have a contradiction. 
        \item \underline{Case 1.3.1.2}: Consider the following equations hold. 
        \begin{align}
            X_1-\alpha_2'+\beta_1'&= p+1 \label{eq:24}\\
            X_2-(\beta_1'+\beta_2')+(\alpha_2'+\gamma_1)&= 1\label{eq:25}
        \end{align}
        If the value of $\beta_2'=\beta_2$, there will be two nodes (node $u_{n-4}$ and node $u_{n-3}$) in $\mathcal{C}_r'$ which have only one agent, which is not a $\mathcal{C}^*$. If the value of $\beta_2'=\beta_2+p$, then due to the following reasoning, we have a contradiction. After taking the sum of Eq. \ref{eq:24} and Eq. \ref{eq:25}, we get the following equation. 
        \begin{align}
             X_1+X_2-\beta_2'+\gamma_1=p+2 \implies X_1+X_2-\beta_2+\gamma_1=2p+2 \text{ as } \beta_2'=\beta_2+p\label{eq:26}
        \end{align}
        Due to Eq. \ref{eq:19} and Eq. \ref{eq:26}, we have the following equation.
        \begin{align}
            X_1+X_2-\beta_2-\alpha_1+\gamma_1=X_1+X_2-\beta_2+\gamma_1\implies \alpha_1=0\label{eq:27}
        \end{align}
        Due to Eq. \ref{eq:11} and Eq. \ref{eq:27}, we have a contradiction. 
         \item \underline{Case 1.3.1.3}: Consider the following equations hold. 
        \begin{align}
            X_1-\alpha_2'+\beta_1'&= p+1 \label{eq:28}\\
            X_2-(\beta_1'+\beta_2')+(\alpha_2'+\gamma_1)&= p+1\label{eq:29}
        \end{align}
       After taking the sum of Eq. \ref{eq:28} and Eq. \ref{eq:29}, we get the following equation. 
        \begin{align}
             X_1+X_2-\beta_2'+\gamma_1=2p+2 \label{eq:30}
        \end{align}
        Due to Eq. \ref{eq:19} and Eq. \ref{eq:30}, we have the following equation.
        \begin{align}
            X_1+X_2-\beta_2-\alpha_1+\gamma_1=X_1+X_2-\beta_2'+\gamma_1\implies \alpha_1=\beta_2'-\beta_2\label{eq:31}
        \end{align}
        In Eq. \ref{eq:31}, if $\beta_2'=\beta_2$ (resp. $\beta_2'=\beta_2+p$), then $\alpha_1=0$ (resp. $\alpha_1=p$). This leads to a contradiction due to Eq. \ref{eq:11}.
        \item \underline{Case 1.3.1.4}: Consider the following equations hold. 
        \begin{align}
            X_1-\alpha_2'+\beta_1'&= 1\\
            X_2-(\beta_1'+\beta_2')+(\alpha_2'+\gamma_1)&= 1
        \end{align}
        This case shows that agents are not in $\mathcal{C}^*$ form at the end of round $ r$. As in $\mathcal{C}^*$, there is exactly one node with one agent.
\end{itemize}

\item \underline{Case 1.3.2}: If the following holds.
\begin{align}
    X_3-(\gamma_1+\gamma_2)+\delta_2+\beta_2&=p+1 \label{eq:32}
\end{align}
In this case, the value of $\beta_2'$ is either $\beta_2$ or $\beta_2-p$ as $n\geq 8$; otherwise the number of agents at $u_{n-4}$ would be neither $p+1$ nor $1$ in $\mathcal{C}_r'$.
\begin{itemize}
    \item \underline{Case 1.3.2.1}: Consider the following equations hold. 
        \begin{align}
            X_1-\alpha_2'+\beta_1'&= 1 \label{eq:33}\\
            X_2-(\beta_1'+\beta_2')+(\alpha_2'+\gamma_1)&= p+1\label{eq:34}
        \end{align}
        If the value of $\beta_2'=\beta_2-p$, there will be two nodes (node $u_{n-4}$ and node $u_{n-2}$) in $\mathcal{C}_r'$ which have only one agent, which is not a $\mathcal{C}^*$. If the value of $\beta_2'=\beta_2$, then due to the following reasoning, we have a contradiction. After taking the sum of Eq. \ref{eq:33} and Eq. \ref{eq:34}, we get the following equation. 
        \begin{align}
             X_1+X_2-\beta_2'+\gamma_1=p+2 \implies X_1+X_2-\beta_2+\gamma_1=p+2 \text{ as } \beta_2'=\beta_2\label{eq:35}
        \end{align}
        Due to Eq. \ref{eq:19} and Eq. \ref{eq:35}, we have the following equation.
        \begin{align}
            X_1+X_2-\beta_2-\alpha_1+\gamma_1=p+ X_1+X_2-\beta_2+\gamma_1\implies \alpha_1=-p\label{eq:36}
        \end{align}
        Due to Eq. \ref{eq:11} and Eq. \ref{eq:36}, we have a contradiction. 
        \item \underline{Case 1.3.2.2}: Consider the following equations hold. 
        \begin{align}
            X_1-\alpha_2'+\beta_1'&= p+1 \label{eq:37}\\
            X_2-(\beta_1'+\beta_2')+(\alpha_2'+\gamma_1)&= 1\label{eq:38}
        \end{align}
        If the value of $\beta_2'=\beta_2-p$, there will be two nodes (node $u_{n-4}$ and node $u_{n-3}$) in $\mathcal{C}_r'$ which have only one agent, which is not a $\mathcal{C}^*$. If the value of $\beta_2'=\beta_2$, then due to the following reasoning, we have a contradiction. After taking the sum of Eq. \ref{eq:37} and Eq. \ref{eq:38}, we get the following equation. 
        \begin{align}
             X_1+X_2-\beta_2'+\gamma_1=p+2 \implies X_1+X_2-\beta_2+\gamma_1=p+2 \text{ as } \beta_2'=\beta_2\label{eq:39}
        \end{align}
        Due to Eq. \ref{eq:19} and Eq. \ref{eq:39}, we have the following equation.
        \begin{align}
            X_1+X_2-\beta_2-\alpha_1+\gamma_1=p+ X_1+X_2-\beta_2+\gamma_1\implies \alpha_1=-p\label{eq:40}
        \end{align}
        Due to Eq. \ref{eq:11} and Eq. \ref{eq:40}, we have a contradiction. 
         \item \underline{Case 1.3.2.3}: Consider the following equations hold. 
        \begin{align}
            X_1-\alpha_2'+\beta_1'&= p+1 \label{eq:41}\\
            X_2-(\beta_1'+\beta_2')+(\alpha_2'+\gamma_1)&= p+1\label{eq:42}
        \end{align}
       After taking the sum of Eq. \ref{eq:41} and Eq. \ref{eq:42}, we get the following equation. 
        \begin{align}
             X_1+X_2-\beta_2'+\gamma_1=2p+2 \label{eq:43}
        \end{align}
        Due to Eq. \ref{eq:19} and Eq. \ref{eq:43}, we have the following equation.
        \begin{align}
            X_1+X_2-\beta_2-\alpha_1+\gamma_1=X_1+X_2-\beta_2'+\gamma_1\implies \alpha_1=\beta_2'-\beta_2\label{eq:44}
        \end{align}
        In Eq. \ref{eq:44}, if $\beta_2'=\beta_2$ (resp. $\beta_2'=\beta_2-p$), then $\alpha_1=0$ (resp. $\alpha_1=-p$). This leads to a contradiction due to Eq. \ref{eq:11}.

        \item  \underline{Case 1.3.2.4}: Consider the following equations hold. 
        \begin{align}
            X_1-\alpha_2'+\beta_1'&= 1\\
            X_2-(\beta_1'+\beta_2')+(\alpha_2'+\gamma_1)&= 1
        \end{align}
        This case shows that agents are not in $\mathcal{C}^*$ form at the end of round $ r$. As in $\mathcal{C}^*$, there is exactly one node with one agent.
\end{itemize}
\end{itemize}
\underline{Case 1.4}: As per pre-computation in $\mathcal{C}_r$, the following equations hold.
    \begin{align}
    X_1-(\alpha_1+\alpha_2)+\beta_1&= 1\\
    X_2-(\beta_1+\beta_2)+(\alpha_2+\gamma_1)&= 1
\end{align}
This case shows that agents are not in $\mathcal{C}^*$ form at the end of round $r$. As in $\mathcal{C}^*$, there is exactly one node with one agent.
\end{itemize}

\vspace{0.1cm}
\begin{figure}
    \centering
    \includegraphics[width=0.75\linewidth]{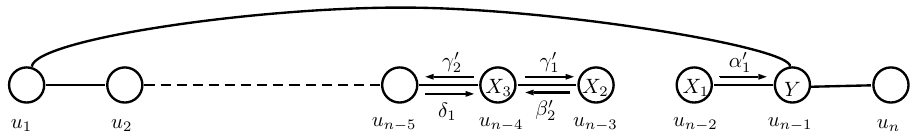}
    \caption{The construction of $\mathcal{C}_r'$ when $Y=1$.}
    \label{fig:1-hop_6}
\end{figure}
\noindent \underline{Case 2 $(Y =1)$}: In this case, $0\leq X_1\leq p$ as if $X_1\geq p+1$, every node $u_i$, $i\in [1,n-3]$ has at least $p+1$. Since agents are not in $\mathcal{C}^*$ configuration at the beginning of round $r$, the total number of agents becomes more than $(n-2)(p+1)+1$, which is not possible. Therefore, the value of $X_1$ belongs to $[0,p]$. The adversary constructs $\mathcal{C}_r'$ as follows. In path $P = u_1 \sim u_2 \sim \ldots \sim u_{n-2} \sim u_{n-1} \sim u_n(=v_n)$, the adversary removes edge between node $u_{n-3}$ and $u_{n-2}$, and add an edge between node $u_1$ and $u_{n-1}$. In configuration $\mathcal{C}_r'$, the adversary deactivates exactly the same agents that were deactivated in $\mathcal{C}_r$. This new graph is $\mathcal{C}_r'$, and constructed as $\mathcal{G}_r$ at the beginning of round $r$. In this graph, the 1-hop view of agents at node $u_1$, $u_2$, $u_{n-4}$, $u_{n-3}$, $u_{n-2}$, $u_{n-1}$ is changed, and the remaining agents at nodes $u_i$, $i\in [3, n-5]$ have the same 1-hop view as in $\mathcal{C}_r$. Since agents are equipped with 1-hop visibility and f-2-f communication, the movement of agents at nodes $u_1$, $u_2$, $u_{n-4}$, $u_{n-3}$, $u_{n-2}$, $u_{n-1}$ in $\mathcal{C}_r'$ may not be the same as $\mathcal{C}_r$, and other remaining agent's movement in $\mathcal{C}_r'$ remains the same as in $\mathcal{C}_r$. Therefore, let $\alpha_1'$ agent(s) move from node $u_{n-2}$ to $u_{n-1}$, from node $u_{n-3}$, let $\beta_2'$ agent(s) move to $u_{n-4}$, and from node $u_{n-4}$, let $\gamma_1'$ agent(s) move to $u_{n-3}$ and $\gamma_2'$ agent(s) move to $u_{n-5}$. We can see this in Fig. \ref{fig:1-hop_6}. 

If $X_1=0$, then no matter how agents move, node $u_{n-2}$ remains a hole as the agent at node $u_{n-1}$ is not active. Therefore, at the end of round $r$, nodes $u_{n-2}$ and $u_n$ remain the hole. The number of agents at node $u_{n-2}$ in $\mathcal{C}_r'$ has to become 1 at the end of round $r$ as node $u_{n-2}$ cannot have $p+1$ agents due to the fact agent at node $u_{n-1}$ is deactivated. Therefore, in $\mathcal{C}_r'$, the following equality should hold.
\begin{align}
    X_2-\beta_2'+\gamma_1'&= p+1\label{eq:49}\\
    X_3-(\gamma_1'+\gamma_2')+\delta_1+\beta_2'&=p+1\label{eq:50}
\end{align}
After taking the sum of Eq. \ref{eq:49} and Eq. \ref{eq:50}, we get the following equation. 
\begin{align}
    X_2+X_3-\gamma_2'+\delta_1=2p+2 \label{eq:51}
\end{align}

\begin{itemize}
    \item \underline{Case 2.1 $(X_1\in [1,p-1])$}: Since one agent is deactivated at node $u_{n-1}$ and $1\leq X_1\leq p-1$, therefore the number of agents at node $u_{n-1}$ in $\mathcal{C}_r$ cannot be $p+1$ at the end of round $r$. If agents achieve $\mathcal{C}^*$ at the end of round $r$ when $\mathcal{G}_r=\mathcal{C}_r$, then the value of $\alpha_1$ is 0, and the following equations hold at the end of round $r$ in $\mathcal{C}_r$ (refer to Figure \ref{fig:1-hop_4}).
\begin{align}
     X_1-\alpha_2+\beta_1&= p+1 \label{eq:45}\\
    X_2-(\beta_1+\beta_2)+(\alpha_2+\gamma_1)&= p+1\label{eq:46}\\
    X_3-(\gamma_1+\gamma_2)+\delta_1+\beta_2&=p+1\label{eq:47}
\end{align}
After taking the sum of Eq. \ref{eq:46} and Eq. \ref{eq:47}, we get the following equation. 
\begin{align}
    X_2+X_3-\gamma_2-\beta_1+\delta_1+\alpha_2=2p+2 \label{eq:48}
\end{align}

In this case, the value of $\gamma_2'=\gamma_2$ as in $\mathcal{C}_r$ (resp $\mathcal{C}_r'$), node $u_{n-1}$ (resp. $u_{n-2}$) becomes a node with 1 agent at the end of round $r$. Therefore, from Eq. \ref{eq:48} and Eq. \ref{eq:51}, we have $\beta_1=\alpha_2$ using $\gamma_2'=\gamma_2$. If we put $\beta_1=\alpha_2$ in Eq. \ref{eq:45}, we get $X_1=p+1$. This leads to a contradiction as $X_1\in [1,p-1]$ in this case. 
    \item \underline{Case 2.2 $(X_1=p)$}: In this case, the value of $\alpha_1=0$ or $p$ due to the following reason. Since agents achieve $\mathcal{C}^*$ at the end of round $r$ and one agent at node $u_{n-1}$ is deactivated, the $E_{\mathcal{C}_r}(u_{n-1})$ is either $1$ or $p+1$ which is possible when $\alpha_1=0$ and $\alpha_1=p$, respectively. 
    \begin{itemize}
        \item \underline{Case 2.2.1 $(\alpha_1=0)$}: In this case, the number of agents at node $u_{n-1}$ in $\mathcal{C}_r$ is 1 at the end of round $r$. If agents achieve $\mathcal{C}^*$ at the end of round $r$ when $\mathcal{G}_r=\mathcal{C}_r$, then the following equations hold at the end of round $r$ in $\mathcal{C}_r$.
        \begin{align}
            X_1-\alpha_2+\beta_1&= p+1 \label{eq:59}\\
            X_2-(\beta_1+\beta_2)+(\alpha_2+\gamma_1)&= p+1\label{eq:60}\\
            X_3-(\gamma_1+\gamma_2)+\delta_1+\beta_2&=p+1\label{eq:61}
        \end{align}
        After taking the sum of Eq. \ref{eq:60} and Eq. \ref{eq:61}, we get the following equation. 
        \begin{align}
            X_2+X_3-\gamma_2-\beta_1+\delta_1+\alpha_2=2p+2 \label{eq:62}
        \end{align}
        
        In this case, the value of $\gamma_2'=\gamma_2$ as in $\mathcal{C}_r$ (resp $\mathcal{C}_r'$), node $u_{n-1}$ (resp. $u_{n-2}$ or $u_{n-1}$) becomes a node with 1 agent at the end of round $r$. Therefore, from Eq. \ref{eq:51} and Eq. \ref{eq:62}, we have $\beta_1=\alpha_2$ using $\gamma_2'=\gamma_2$. If we put $\beta_1=\alpha_2$ in \label{eq:59}, we get $X_1=p+1$. This leads to a contradiction as $X_1=p$ in this case. 
        \item \underline{Case 2.2.2 $(\alpha_1=p)$}: In this case, the number of agents at node $u_{n-1}$ in $\mathcal{C}_r$ is $p+1$ at the end of round $r$. Since $X_1=p$ and $\alpha_1=p$, the value of $\alpha_2=0$. If agents achieve $\mathcal{C}^*$ at the end of round $r$ when $\mathcal{G}_r=\mathcal{C}_r$, then the following equations hold at the end of round $r$ in $\mathcal{C}_r$.
        \begin{align}
            X_1-(\alpha_1+\alpha_2)+\beta_1&= 1 \text{ or }p+1 \label{eq:66}\\
            X_2-(\beta_1+\beta_2)+(\alpha_2+\gamma_1)&= 1 \text{ or } p+1\label{eq:67}\\
            X_3-(\gamma_1+\gamma_2)+\delta_1+\beta_2&=1 \text{ or }p+1\label{eq:68}
        \end{align}
        
        \begin{itemize}
            \item \underline{Case 2.2.2.1}: Consider the following holds:
            \begin{align}
                X_1-(\alpha_1+\alpha_2)+\beta_1&= 1\implies \beta_1=1 \text{ as } X_1=\alpha_1=p, \alpha_2=0\label{eq:72}
            \end{align}
            Therefore, the following equation holds in $\mathcal{C}_r$ at the end of round $r$.
            \begin{align}
                X_2-(\beta_1+\beta_2)+(\alpha_2+\gamma_1)&= p+1\label{eq:73}\\
                X_3-(\gamma_1+\gamma_2)+\delta_1+\beta_2&=p+1\label{eq:74}
            \end{align}      
            After taking the sum of Eq. \ref{eq:73} and Eq. \ref{eq:74}, we get the following equation using $\alpha_2=0$. 
            \begin{align}
                X_2+X_3-\gamma_2-\beta_1+\delta_1=2p+2 \label{eq:75}
            \end{align}
            In this case, the value of $\gamma_2'=\gamma_2$ in $\mathcal{C}'_r$ as the number of agents at node $u_{n-5}$ in $\mathcal{C}_r$ and $\mathcal{C}_r'$ is $p+1$; otherwise there will be two nodes with one agent in $\mathcal{C}_r'$. Using Eq. \ref{eq:51} and Eq. \ref{eq:75}, we have $\beta_1=0$, which leads to a contradiction due to Eq. \ref{eq:72}. 
            \item \underline{Case 2.2.2.2}: Consider the following holds:
            \begin{align}
                X_1-(\alpha_1+\alpha_2)+\beta_1&= p+1\implies \beta_1=p+1 \text{ as } X_1=\alpha_1=p, \alpha_2=0\label{eq:76}
            \end{align}
            Therefore, the following equation holds in $\mathcal{C}_r$ at the end of round $r$.
            \begin{align*}
                X_2-(\beta_1+\beta_2)+(\alpha_2+\gamma_1)&= 1 \text{ or } p+1\\
                X_3-(\gamma_1+\gamma_2)+\delta_1+\beta_2&=1 \text{ or } p+1
            \end{align*}      

            \noindent $\rightarrow$ \underline{Case 2.2.2.2.1}: Consider the following holds.
                 \begin{align}
                    X_2-(\beta_1+\beta_2)+(\alpha_2+\gamma_1)&= 1 \label{eq:77}\\
                    X_3-(\gamma_1+\gamma_2)+\delta_1+\beta_2&=p+1\label{eq:78}
                \end{align}
                After taking the sum of Eq. \ref{eq:77} and Eq. \ref{eq:78}, we have the following using $\alpha_2=0$.
                \begin{align}
                X_2+X_3-\gamma_2-\beta_1+\delta_1=p+2 \label{eq:79}
            \end{align}
            In this case, the value of $\gamma_2'=\gamma_2$ in $\mathcal{C}'_r$ as the number of agents at node $u_{n-5}$ in $\mathcal{C}_r$ and $\mathcal{C}_r'$ is $p+1$; otherwise there will be two nodes with one agent in $\mathcal{C}_r'$. Therefore, using Eq. \ref{eq:51} and Eq. \ref{eq:79}, we have $\beta_1=p$, which leads to a contradiction due to Eq. \ref{eq:76}. 
            
            \noindent $\rightarrow$ \underline{Case 2.2.2.2.2}: Consider the following holds.
                 \begin{align}
                    X_2-(\beta_1+\beta_2)+(\alpha_2+\gamma_1)&= p+1 \label{eq:80}\\
                    X_3-(\gamma_1+\gamma_2)+\delta_1+\beta_2&=1\label{eq:81}
                \end{align}
                After taking the sum of Eq. \ref{eq:80} and Eq. \ref{eq:81}, we have the following.
                \begin{align}
                X_2+X_3-\gamma_2-\beta_1+\delta_1=p+2 \label{eq:82}
            \end{align}
            In this case, the value of $\gamma_2'=\gamma_2$ in $\mathcal{C}'_r$ as the number of agents at node $u_{n-4}$ in $\mathcal{C}_r$ and $\mathcal{C}_r'$ is $p+1$; otherwise there will be two nodes with one agent in $\mathcal{C}_r'$. Therefore, using Eq. \ref{eq:51} and Eq. \ref{eq:82}, we have $\beta_1=p$, which leads to a contradiction due to Eq. \ref{eq:76}.
            
            \noindent $\rightarrow$ \underline{Case 2.2.2.2.3}: Consider the following holds.
                 \begin{align}
                    X_2-(\beta_1+\beta_2)+(\alpha_2+\gamma_1)&= p+1 \label{eq:83}\\
                    X_3-(\gamma_1+\gamma_2)+\delta_1+\beta_2&=p+1\label{eq:84}
                \end{align}
                After taking the sum of Eq. \ref{eq:83} and Eq. \ref{eq:84}, we have the following using $\alpha_2=0$.
                \begin{align}
                X_2+X_3-\gamma_2-\beta_1+\delta_1=2p+2 \label{eq:85}
            \end{align}

            In this case, the number of agents at node $u_{n-5}$ at the end of round $r$ is either 1 or $p+1$ in $\mathcal{C}_r$. Since at the end of round $r$, the number of agents at node $u_{n-2}$ is 1 in $\mathcal{C}_r'$, the number of agents at node $u_{n-5}$ has to be $p+1$ at the end of round $r$ in $\mathcal{C}_r'$. Hence, $\gamma_2'=\gamma_2$ (resp. $\gamma_2'=\gamma_2+p$) if the number of agents at the node $u_{n-5}$ at the end of round $r$ is $p+1$ (resp. 1) in $\mathcal{C}_r$. In both cases, the value of $\beta_1\neq p+1$ using Eq. \ref{eq:51} and Eq. \ref{eq:85} leads to a contradiction due to Eq. \ref{eq:76}
    \end{itemize} 
\end{itemize}
\end{itemize}

\vspace{0.1cm}
\begin{figure}
    \centering
    \includegraphics[width=0.75\linewidth]{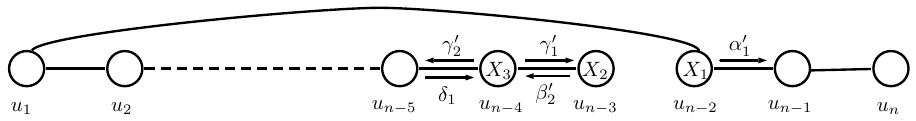}
    \caption{The construction of $\mathcal{C}_r'$ when $Y=0$.}
    \label{fig:1-hop_7}
\end{figure}
\noindent \underline{Case 3 $(Y =0)$}: In this case, $X_1\geq p+2$ is not possible as this leads to at least $(n-2)(p+2)$ agents in the system but we have total $(p+1)(n-2)+1$ many agents. Also, $X_1=0$ is not possible as no matter how agents move in $\mathcal{C}_r$, there will be at least two holes as all agents are deactivated at node $u_{n-2}$. Therefore, the value of $X_1$ is in $[1,p+1]$. If the value of $X_1$ is from $[1,p]$, and agents achieve $\mathcal{C}^*$ at the end of round $r$ in $\mathcal{C}_r$, then node $u_{n-1}$ has to become a node with one agent at the end of round $r$, i.e., $\alpha_1=1$. If the value of $X_1$ is $p+1$, then the adversary deactivates $1$ agent at node $u_{n-2}$ in $\mathcal{C}_r$ (which was not deactivated earlier). In this case, if agents achieve $\mathcal{C}^*$ at the end of round $r$ in $\mathcal{C}_r$, node $u_{n-1}$ has to become a node with one agent, i.e., $\alpha_1=1$. Therefore, the following holds in $\mathcal{C}_r$ at the end of round $r$ when $1\leq X_1\leq p+1$.
\begin{align}
     X_1-(\alpha_1+\alpha_2)+\beta_1&= p+1 \label{eq:52}\\
    X_2-(\beta_1+\beta_2)+(\alpha_2+\gamma_1)&= p+1\label{eq:53}\\
    X_3-(\gamma_1+\gamma_2)+\delta_1+\beta_2&=p+1\label{eq:54}
\end{align}

After taking the sum of Eq. \ref{eq:53} and Eq. \ref{eq:54}, we get the following equation. 
\begin{align}
    X_2+X_3-\gamma_2-\beta_1+\delta_1=2p+2 \label{eq:55}
\end{align}

The adversary constructs $\mathcal{C}_r'$ as follows. In path $P = u_1 \sim u_2 \sim \ldots \sim u_{n-2} \sim u_{n-1} \sim u_n(=v_n)$, the adversary removes edge between node $u_{n-3}$ and $u_{n-2}$, and add an edge between node $u_1$ and $u_{n-2}$. In configuration $\mathcal{C}_r'$, the adversary deactivates exactly the same agents that were deactivated in $\mathcal{C}_r$. This new graph is $\mathcal{C}_r'$, and constructed as $\mathcal{G}_r$ at the beginning of round $r$. In this graph, the 1-hop view of agents at nodes $u_1$, $u_2$, $u_{n-4}$, $u_{n-3}$, $u_{n-2}$, $u_{n-1}$ are changed, and the remaining agents at nodes $u_i$, $i\in [3, n-5]$ have the same 1-hop view as in $\mathcal{C}_r$. Since agents are equipped with 1-hop visibility and f-2-f communication, the movement of agents at nodes $u_1$, $u_2$, $u_{n-4}$, $u_{n-3}$, $u_{n-2}$, $u_{n-1}$ in $\mathcal{C}_r'$ may not be the same as $\mathcal{C}_r$, and other remaining agent's movement in $\mathcal{C}_r'$ remains the same as in $\mathcal{C}_r$. Therefore, let $\alpha_1'$ agent(s) move from node $u_{n-2}$ to $u_{n-1}$, from node $u_{n-3}$, let $\beta_2'$ agent(s) move to $u_{n-4}$, and from node $u_{n-4}$, let $\gamma_1'$ agent(s) move to $u_{n-3}$ and $\gamma_2'$ agent(s) move to $u_{n-5}$. We can see this in Fig. \ref{fig:1-hop_7}.

In $\mathcal{C}_r'$, the number of agents at node $u_{n-1}$ has to become 1; otherwise, at the end of round $r$ in $\mathcal{C}_r'$, node $u_{n-1}$ is either a hole or has agents between 2 and $p$. Therefore, the following inequality must hold.

\begin{align}
    X_2-\beta_2'+\gamma_1'&= p+1\label{eq:56}\\
    X_3-(\gamma_1'+\gamma_2')+\delta_1+\beta_2'&=p+1\label{eq:57}
\end{align}
After taking the sum of Eq. \ref{eq:56} and Eq. \ref{eq:57}, we get the following equation. 
\begin{align}
    X_2+X_3-\gamma_2'+\delta_1=2p+2 \label{eq:58}
\end{align}
In this case, the value of $\gamma_2'=\gamma_2$ as in $\mathcal{C}_r$ and $\mathcal{C}_r'$, node $u_{n-1}$ becomes a node with 1 agent at the end of round $r$. Therefore, from Eq. \ref{eq:55} and Eq. \ref{eq:58}, we have $\beta_1=0$ using $\gamma_2'=\gamma_2$. Since $\beta_1=0$, $\alpha_1=1$ and $X_1\in [1,p+1]$, therefore, $X_1-(\alpha_1+\alpha_2)+\beta_2\leq p$. This leads to a contradiction due to Eq. \ref{eq:52}.

\vspace{0.2cm}
\noindent\textbf{Fairness.} The set of deactivated agents remains identical in $\mathcal{C}_r$ and $\mathcal{C}_r'$ for every $r \geq 0$. Moreover, $\mathcal{G}_r$ is defined to be either $\mathcal{C}_r$ or $\mathcal{C}r'$. By construction, when forming $\mathcal{C}{r+1}$ from $\mathcal{G}_r$, every agent that was deactivated in round $r$ becomes activated in round $r+1$. Therefore, no agent remains deactivated forever, which ensures fairness.

\medskip
It is clear to observe that in construction $\mathcal{C}_r$ and $\mathcal{C}_r'$, the degree of node $w_n(=v_n)$ remains one. Further, since the proof does not rely on any limitation on node storage, agent memory, or agent knowledge, it remains valid even if nodes are equipped with unbounded storage and agents have unbounded memory and complete parameter knowledge. This completes the proof.
\end{proof}

\begin{theorem}\label{thm:imp-global}
($n\geq 3, p\geq 1$) If the adversary deactivates at most $p$ agents per round, exploration is impossible with $(n-2)(p+1)+1$ agents in 1-interval connected graphs with $0$-hop visibility and global communication, even if nodes have unique identifiers and unbounded storage, and agents possess unbounded memory and complete parameter knowledge.
\end{theorem}
\begin{proof}
   Let $G$ be a clique of size $n$. Consider any initial configuration with at least one hole. Label the nodes $v_1,\dots,v_n$ such that $\alpha(v_i)\ge \alpha(v_{i+1})$ for all $i\in[1,n-1]$, where $\alpha(v_j)$ denotes the initial number of agents at $v_j$. We construct $\mathcal{G}_r$ such that $v_n$ remains a hole forever. We use two notations. Let $H_r$ be a graph at round $r$. For any node $v \in H_r$, let $S_{H_r}(v)$ and $E_{H_r}(u)$ denote the number of agents at the beginning and end of round $r$ in $H_r$, respectively.

       \vspace{0.1cm}
\noindent \textbf{Graph $\mathcal{C}_0$:} Consider a path $P_0:=v_1\sim v_2\sim v_3\sim \ldots v_{n-1}\sim v_n$. In this case, $\mathcal{C}_0=P_0$. The adversary deactivates the agents as follows.
\begin{itemize}
    \item If there are at least two agents at node $v_{n-1}$, the adversary deactivates all agents at node $v_{n-1}$, and all other agents are active. The adversary can deactivate all agents, since node $v_{n-1}$ can contain at most $p$ agents. If not, $\alpha(v_i)\geq p+1$, $i\in [1,n-1]$, as $\alpha(v_i)\ge \alpha(v_{i+1})$ for every $i\in[1,n-1]$. In this case, the total number of agents is at least $(n-1)(p+1)$, which is more than $(n-2)(p+1)+1$.
    \item If there is at most one agent at node $v_{n-1}$, it does not deactivate any agent in $\mathcal{C}_0$.
\end{itemize}

    \vspace{0.15cm}
\noindent \textbf{Graph $\mathcal{C}_r$ at round $r\geq 1$:} At round $r-1$, let $w_1$, $w_2$, \ldots,$w_{n-1}$, $w_n(=v_n)$ be nodes in $\mathcal{G}_{r-1}$. Note that in $\mathcal{G}_{r-1}$, the value of $deg_r(w_n)$ is 1 (we observe this fact after completing the construction of $\mathcal{G}_r$). Without loss of generality, let $w_{n-1}$ be the neighbour of $w_n$ in $\mathcal{G}_{r-1}$. Let $w_1', w_2', \ldots, w_n'$ be an ordering of the nodes such that $\bigcup_{i=1}^n \{w_i'\} = \bigcup_{i=1}^n \{w_i\}$, $w_n'=w_n$ and $E_{\mathcal{G}_{r-1}}(w_i') \ge E_{\mathcal{G}_{r-1}}(w_{i+1}')$ for every $i \in [1, n-1]$. There are two possible cases: 
\begin{itemize}
    \item \textbf{Case 1} \textbf{$\bm{\big(\exists \;w\in \{w_1', w_2', \ldots, w_{n-1}'\}}$ such that $\bm{E_{\mathcal{G}_{r-1}}(w)\leq 1\big)}$}: In this case, $E_{\mathcal{G}_{r-1}}(w_{n-1}')\leq 1$. Consider a path, $P_r:=w_1'\sim w_2'\sim w_3'\sim \ldots w_{n-1}'\sim w_n'(=w_n=v_n)$. In this case, $\mathcal{C}_r=P_r$, and the adversary does not deactivate any agent in $\mathcal{C}_r$.
    \item \textbf{Case 2} \textbf{$\bm{\big(\nexists \;w\in \{w_1', w_2', \ldots, w_{n-1}'\}}$ such that $\bm{E_{\mathcal{G}_{r-1}}(w)\leq 1\big)}$}: The value of $E_{\mathcal{G}_{r-1}}(w'_{n-1})$ is at most $p$. If not, we have $E_{\mathcal{G}_{r-1}}(w_i') \ge p+1$ for every $i \in [1, n-1]$, implying that the total number of agents is at least $(p+1)(n-1)$, which is greater than $(p+1)(n-2)+1$. Therefore, $E_{\mathcal{G}_{r-1}}(w'_{n-1})\in [2,p]$. Since $E_{\mathcal{G}_{r-1}}(w'_{n-1})\in [2,p]$, the value of $E_{\mathcal{G}_{r-1}}(w_{n-2}')$ is at most $p$. If $w_{n-1}' \neq w_{n-1}$, then the adversary defines the path
$
P_r := w_1' \sim w_2' \sim w_3' \sim \cdots \sim w_{n-1}' \sim w_n'(=w_n=v_n),
$
and deactivates all agents at node $w_{n-1}'$, while keeping all other agents active. Otherwise, if $w_{n-1}' = w_{n-1}$, the adversary defines
$
P_r := w_1' \sim w_2' \sim w_3' \sim \cdots \sim w_{n-1}' \sim w_{n-2}' \sim w_n'(=w_n=v_n),
$ and deactivates all agents at node $w_{n-2}'$, keeping all other agents active. In both cases, $\mathcal{C}_r = P_r$.
\end{itemize}

\medskip
    As per the construction of $\mathcal{C}_r$ for round $r\ge 0$, the graph is a path $P = u_1 \sim u_2 \sim \ldots \sim u_{n-2} \sim u_{n-1} \sim u_n(=v_n)$. The values $S_{\mathcal{C}_r}(u_i)$ do not increase as we move from $u_i$ to $u_{i+1}$ for every $i\in[1,n-3]$, and $S_{\mathcal{C}_r}(u_{n-1})\leq p$. Since the adversary knows the agents' algorithm and the structure of $\mathcal{C}_r$, it can pre-compute the agents' movement and determine whether one of the agents visits node $v_n$ when $\mathcal{G}_r=\mathcal{C}_r$ at the end of round $r$. If not, the adversary simply forms $\mathcal{G}_r=\mathcal{C}_r$ at the beginning of round $r$. Otherwise, the adversary instead forms $\mathcal{C}_r'$ at the beginning of round $r$. We show that whenever one of the agents visits node $v_n$ under $\mathcal{C}_r$, node $v_n$ is not visited by any agent under $\mathcal{C}_r'$ at the end of round $r$. The construction of $\mathcal{C}_r'$ depends on the value of $S_{\mathcal{C}_r}(u_{\,n-1})$. Let $S_{\mathcal{C}_r}(u_{n-1}) = Y$. Since $S_{\mathcal{C}_r}(u_{n-1}) \le p$, it follows that $Y \in [0,p]$.

\begin{figure}[h]
\begin{minipage}[b]{0.48\linewidth}
 \centering
    \includegraphics[width=0.7\linewidth]{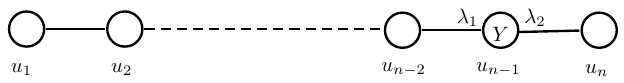}
    \caption{The construction of $\mathcal{C}_r$ when $Y=1$.}
    \label{fig:global_4}
 \end{minipage}
 \begin{minipage}[b]{0.48\linewidth}
 \centering
    \includegraphics[width=0.7\linewidth]{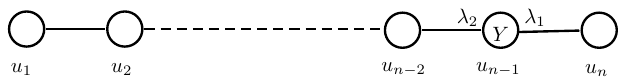}
    \caption{The construction of $\mathcal{C}_r'$ when $Y=1$.}
    \label{fig:global_5}
 \end{minipage}
 \end{figure}

 \begin{itemize}
     \item \underline{$Y\in [0,p]\setminus\{1\}$:} Since either no agent is present at node $u_{n-1}$ or all agents at node $u_{n-1}$ are not active in configuration $\mathcal{C}_r$, node $u_n$ is at a distance of at least two from every active agent in $\mathcal{C}_r$.
     \item \underline{$Y = 1$:} In this case, all agents are active. If node $u_n$ is to be visited by some agent, then the agent located at node $u_{n-1}$ must move to $u_n$. Let $\lambda_1$ denote the outgoing port at node $u_{n-1}$ leading to node $u_{n-2}$, and let $\lambda_2$ denote the outgoing port at node $u_{n-1}$ leading to node $u_n$ (refer to Fig. \ref{fig:global_4}). We now describe the construction of the configuration $\mathcal{C}_r'$. Agents are equipped with global communication and $0$-hop visibility; hence, an agent can sense only the outgoing ports of the node it currently occupies. The adversary constructs $\mathcal{C}_r'$ so that the $0$-hop visibility view of every agent in $\mathcal{C}_r'$ is identical to that in $\mathcal{C}_r$. Consequently, even with global communication, agents cannot detect that the configuration has changed. Therefore, any port choice made by an agent in $\mathcal{C}_r$ is also made in $\mathcal{C}_r'$. Let $P_r' = u_1 \sim u_2 \sim \cdots \sim u_{n-2} \sim u_{n-1} \sim u_n (= v_n),$ where, in $\mathcal{C}_r'$, the outgoing port $\lambda_1$ at node $u_{n-1}$ leads to node $u_n$, and the outgoing port $\lambda_2$ leads to node $u_{n-2}$ (refer to Fig. \ref{fig:global_5}). Thus, $\mathcal{C}_r'$ is the path $P_r'$, and the adversary forms $\mathcal{C}_r'$ as the snapshot $\mathcal{G}_r$ instead of $\mathcal{C}_r$. Since the $0$-hop visibility of every agent is identical in $\mathcal{C}_r$ and $\mathcal{C}_r'$, the agents compute the same port choices in both configurations. As a result, if an agent at node $u_{n-1}$ selects port $\lambda_2$ in $\mathcal{C}_r$ to move to node $u_n$, it selects the same port in $\mathcal{C}_r'$ and instead moves to node $u_{n-2}$. Consequently, node $u_n$ is not visited by any active agent at the end of round $r$. Since this argument is valid for every $r\geq 0$, no active agent visits node $v_n$ at any round $r$. 
 \end{itemize}
 
\medskip
\noindent\textbf{Fairness.} The way we construct $\mathcal{C}_r$ from $\mathcal{G}_{r-1}$, it is trivial that those agents which were deactivated in round $r-1$ are activated in round $r$. This guarantees fairness. 
       
Since the proof does not rely on any limitations except $0$-hop visibility and global communication, it remains valid even if nodes have unbounded storage and unique identifiers, and agents have unbounded memory and complete parameter knowledge. This completes the proof. 
\end{proof}

Due to Theorem \ref{thm:imp-1-hop} and Theorem \ref{thm:imp-global}, we have the following observation.

\begin{observation}\label{obs:necessary}
For $1$-interval connected graphs on $n$ nodes, exploration against an adversary that can deactivate at most $p$ agents per round might be possible using $(n-2)(p+1)+1$ agents when agents have $1$-hop visibility and $1$-hop communication. Consequently, if exploration is solvable with $k$ agents under an adversary deactivating at most $p$ agents per round, then
$k \ge (p+1)(n-2)+1$, which implies
$p \le \left\lfloor \frac{k-1}{n-2} \right\rfloor - 1$.
\end{observation}

\section{Exploration in 1-interval connected graphs}\label{sec:ssync_algo}
In this section, we present {\ssync}\_$\mathcal{EXPO}$ for exploration in $1$-interval connected graphs under the {\ssync} scheduler. The high-level idea of our approach is in Section~\ref{sec:tech}. Let $\mathcal{G}_r=(V,E(r))$ denote the graph at round $r$. In the graph, there are $k$ agents, and the adversary can deactivate at most $p=\left\lfloor \frac{k-1}{n-2} \right\rfloor - 1$ agents per round. Agents are equipped with $1$-hop visibility and global communication\footnote{A discussion on the technical difficulty with weaker communication is discussed in the conclusion (refer to Section~\ref{sec:con}).}. Based on the active agents in $\mathcal{G}_r$, we introduce the following notion.

\vspace{-0.2cm}

\begin{definition}\label{def:1}
\textbf{(Connected components of active agents in $\mathcal{G}_r$)}
The graph $\mathcal{G}_r$ can be partitioned into subgraphs $\mathcal{G}_r^1, \mathcal{G}_r^2, \ldots, \mathcal{G}_r^{\ell}$, where each $\mathcal{G}_r^j = (V_j, E_j)$ for $1 \leq j \leq {\ell}$, such that: (i) For every $u \in V_j$, at least one active agent is present at $u$, (ii) For all $j \neq i$, $V_j \cap V_i = \emptyset$, and (iii) There is no edge $e = (u_1,u_2) \in E(r)$ with $u_1 \in V_j$ and $u_2 \in V_i$ for $j \neq i$.
\end{definition}

We denote by $CCA(\mathcal{G}_r)$ the collection of all such subgraphs. Intuitively, $CCA(\mathcal{G}_r)$ represents the connected components of the subgraph of $\mathcal{G}_r$ induced by nodes that host at least one active agent. 

\vspace{0.1cm}
\noindent \textbf{Notation and parameters:} Each agent $\mathcal{A}$ maintains the following parameters. The variable $\mathcal{A}.\text{ID}$ denotes the unique identifier of agent $\mathcal{A}$. The counter $\mathcal{A}.p$ stores the number of inactive agents observed by $\mathcal{A}$; it takes non-negative integer values and is initially set to $0$. The boolean variable $\mathcal{A}.\text{EXP}$ indicates whether exploration has been achieved, where value $1$ denotes success and $0$ otherwise; initially, $\mathcal{A}.\text{EXP}=0$. The variable $\mathcal{A}.\text{srcID}$ stores the identifier currently considered as the source and is initially set to $\bot$. The boolean variable $\mathcal{A}.\text{phase}$ indicates whether the source identifier has been fixed, where $0$ denotes not fixed, and $1$ denotes fixed; initially, $\mathcal{A}.\text{phase}=0$. For any two adjacent nodes $v_1$ and $v_2$, $\pi(v_1,v_2)$ denotes the port label at node $v_1$ leading to its neighbor $v_2$. Finally, for any set $S$, we denote by $|S|$ the cardinality of $S$, that is, the number of elements in $S$.

\vspace{0.1cm}
The execution of {\ssync}\_$\mathcal{EXPO}$ proceeds in three phases in each round $r$. Let an active agent $\mathcal{A}$ be located at node $v \in \mathcal{G}_r$. In round $r$, an active agent $\mathcal{A}$ executes these three phases.

\medskip
\noindent \textbf{Phase 1 (1-hop view collection)}: For each port $\overline{q}_1 \in \{0,1,\ldots,\deg_r(v)-1\}$, agent $\mathcal{A}$ performs the following steps.
\begin{itemize}
    \item Let $u$ be the neighbor of $v$ reachable via port $\overline{q}_1$. Let $ID(v)$ and $ID(u)$ denote the sets of identifiers of agents located at $v$ and $u$, respectively.
    
    \item If $ID(u) \neq \emptyset$, then define $C_v^{\overline{q}_1} = (ID(v), \overline{q}_1, ID(u))$.
    
    \item Otherwise, if $ID(u) = \emptyset$, define $C_v^{\overline{q}_1} = (ID(v), \overline{q}_1, \emptyset)$., where $\emptyset$ indicates that $u$ is a hole.
\end{itemize}

Let $C_v = (C_v^0, C_v^1, \ldots, C_v^{d})$ denote the 1-hop view of agent $\mathcal{A}$ at node $v$, where $d = \deg_r(v)-1$. And, let $\text{Info}(\mathcal{A})=(\mathcal{A}.ID, \mathcal{A}.p, \mathcal{A}.\text{EXP},\mathcal{A}.\text{srcID}, \mathcal{A}.\text{phase})$. Agent $\mathcal{A}$ broadcasts $\big(\text{Info}(\mathcal{A}), C_v\big)$.

\medskip
\noindent \textbf{Phase 2 (Graph reconstruction)}: Let $v_1, v_2, \ldots, v_{\ell}$ be nodes with at least one active agent at round $r$, and note that $v \in \bigcup_{i=1}^{\ell}\{v_i\}$. Since each active agent broadcasts its 1-hop view, the 1-hop view of a node $v_i$ may be broadcast by multiple agents in the same round. Let $b_1^i, b_2^i, \ldots, b_{\ell_i}^i$ be the active agents present at node $v_i$, for every $i \in [1,\ell]$. Each agent $b_j^i$ at node $v_i$ broadcasts $\big(\text{Info}(b_j^i), C_{v_i}\big)$. From $C_{v_i}$, the set $ID(v_i)$ can be extracted, and using $ID(v_i)$ together with the identifiers of agents $b_j^i$, the minimum active agent at node $v_i$ can be determined; without loss of generality, let $b_1^i$ be this agent. Let $x_i = |ID(v_i)| - \ell_i$ denote the number of inactive agents at node $v_i$. If there exists a port $q$ at node $v_i$ leading to a node $w$ with at least one agent, then $ID(w)$ can be obtained from $C_{v_i}$; if no agent from $w$ broadcasts its 1-hop view, it follows that all agents at $w$ are inactive. Let $w_1, w_2, \ldots, w_{\tau}$ be such nodes which have only inactive agents, identifiable from $\bigcup_{i=1}^{\ell}\{C_{v_i}\}$. Define $\mathcal{X} = \sum_{j=1}^{\tau} |ID(w_j)| + \sum_{i=1}^{\ell} x_i$, and $\mathcal{Y} = \max\{ b_j^i.p \mid j \in [1,\ell_i],\, i \in [1,\ell] \}$. Here, $\mathcal{X}$ represents the number of inactive agents at round $r$ which active agents can observe, and $\mathcal{Y}$ is the maximum $p$ value among all active agents. Based on the gathered knowledge using global communication, it can define the above notations. It forms the map $G'=(V',E')$ as follows.

\begin{itemize}
    \item The node set is defined as $V' = \left\{ ID_{v_i} \mid i \in [1,{\ell}]\right\}$, where $ID_{v_i}$ is the ID of minimum active agent (i.e., $b_1^i.ID$) at node $v_i$. 

    \item The edge set $E'$ is constructed as follows. For every pair of tuples $(ID(u_1), q_1, ID(u_2))$ and $(ID(u_2), q_2, ID(u_1))$, agent $\mathcal{A}$ adds an undirected edge $(ID_{u_1}, ID_{u_2})$ with port labels $\pi(ID_{u_1}, ID_{u_2}) = q_1$ and $\pi(ID_{u_2}, ID_{u_1}) = q_2$.

    \item For each tuple $(ID(u_1), q_1, ID(u_2))$, if no view is received from any agent at $u_2$, agent $\mathcal{A}$ marks port $q_1$ at node $ID_{u_1}$ as leading to a node containing only inactive agents. It stores set $ID(u_2)$ corresponding port $q_1$.

    \item For each tuple $(ID(u_1), q_1, \emptyset)$, agent $\mathcal{A}$ marks port $q_1$ at node $ID_{u_1}$ leads to a hole.
\end{itemize}

\medskip
\noindent \textbf{Phase 3 (Move)}: In the second phase, agent $\mathcal{A}$ constructs a map $G'$. If there is an agent $b_j^i$ at node $v_i$ such that $b_j^i.\text{EXP}=1$, it sets $\mathcal{A}.\text{EXP}=1$. Otherwise, it does the following. If $G'$ is disconnected, it considers the connected component $G''$ of $G'$ in which node $v$ (location of agent $\mathcal{A}$ at round $r$) is present. There are two cases based on $b_j^i.\text{phase}=1$, which are as follows.

\vspace{0.1cm}
\noindent \textbf{Case 1 (there is no agent $b_j^i$ with $b_j^i.\text{phase}=1$):} It sets $\mathcal{A}.p=\max\{\mathcal{X}, \mathcal{Y}\}$. There are two sub-cases which are as follows.

\begin{itemize}
    \item \textbf{Sub-case 1 (at least one node in $G'$ has $\mathcal{A}.p + 2$ or more agents)}: If agent $\mathcal{A}$ is not the minimum active agent at node $v$, it stays idle; otherwise, based on $G'$, it proceeds as follows. Let $\overline{w}_1, \overline{w}_2, \ldots, \overline{w}_{\lambda_1}$ be the nodes in $G'$ that have at least $\mathcal{A}.p + 2$ agents. Let $\overline{b}_j$ be the minimum ID active agent at node $\overline{w}_j$. Without loss of generality, let $\overline{b}_1.ID = \min\{ \overline{b}_j.ID \mid j \in [1,\lambda_1] \}$. Consider the active connected component (say $H$) of $G'$ where node $\overline{w}_1$ is present. If $H\neq G''$, agent $\mathcal{A}$ stays idle at round $r$. Otherwise, it does the following. 

    \vspace{0.1cm}
     \noindent \underline{A port in $G''$ that leads to a hole:} Let $\overline{w}'_1, \overline{w}'_2, \ldots, \overline{w}'_{\lambda_2}$ be the nodes in $G''$ that have at least one port leading to a hole. Let $\overline{b}'_j$ be the minimum ID active agent at node $\overline{w}'_j$. Without loss of generality, let $\overline{b}'_1.ID = \min\{ \overline{b}'_j.ID \mid j \in [1,\lambda_2] \}$. Since agent $\mathcal{A}$ is aware of $G''$, it considers a shortest path $P$ between $\overline{w}_1$ and $\overline{w}'_1$. If there are multiple shortest paths between $\overline{w}_1$ and $\overline{w}'_1$, it selects the lexicographically smallest among them. Let $P = \overline{u}_1(=\overline{w}_1) \sim \overline{u}_2 \sim \ldots \sim \overline{u}_{\lambda}(=\overline{w}'_1)$ be this lexicographically shortest path in $G''$. If $v = \overline{u}_j$ for some $1 \leq j < \lambda$, agent $\mathcal{A}$ moves to node $\overline{u}_{j+1}$. Otherwise, if agent $\mathcal{A}$ is at node $\overline{u}_{\lambda}$, it moves through the minimum available port that leads to a hole. In all other cases, agent $\mathcal{A}$ remains at node $v$. This is the \text{pipeline} approach.

            \vspace{0.15cm}
         \noindent \underline{No port in $G''$ that leads to a hole:} If no port in $G'$ leads to a node which has only inactive agents, it updates $\mathcal{A}.EXP=1$. Otherwise, it does the following. Let $\overline{w}'_1, \overline{w}'_2, \ldots, \overline{w}'_{\lambda_2}$ be the nodes in $G''$ that have at least one port leading to a node where all agents are inactive. Let $\overline{b}'_j$ be the minimum ID active agent at node $\overline{w}'_j$. Without loss of generality, let $\overline{b}'_1.ID = \min\{ \overline{b}'_j.ID \mid j \in [1,\lambda_2] \}$. Since agent $\mathcal{A}$ is aware of $G''$, it considers a shortest path $P$ between $\overline{w}_1$ and $\overline{w}'_1$. If there are multiple shortest paths between $\overline{w}_1$ and $\overline{w}'_1$, it selects the lexicographically smallest among them. Let $P = \overline{u}_1(=\overline{w}_1) \sim \overline{u}_2 \sim \ldots \sim \overline{u}_{\lambda}(=\overline{w}'_1)$ be this lexicographically shortest path in $G''$. If $v = \overline{u}_j$ for some $1 \leq j < \lambda$, agent $\mathcal{A}$ moves to node $\overline{u}_{j+1}$. Otherwise, if agent $\mathcal{A}$ is at node $\overline{u}_{\lambda}$, it moves through the minimum available port that leads to a node which has only inactive agent(s). In all other cases, agent $\mathcal{A}$ remains at node $v$.
         \item  \textbf{Sub-case 2 (no node in $G'$ has $\mathcal{A}.p + 2$ or more agents)}: In the analysis, we show that at most one hole remains (refer to observation \ref{obs:neg2hole}). If only one hole exists, then $\tau$, the number of nodes containing only inactive agents that are visible to at least one active agent, is at most $1$ (refer to Lemma \ref{lm:key}). Agent $\mathcal{A}$ updates $\mathcal{A}.\text{phase} = 1$. Further, if agent $\mathcal{A}$ finds that $\tau \neq 1$, it updates $\mathcal{A}.\text{EXP} = 1$. Otherwise, let
$x = \min \big\{ \min\{ |ID(v_i)|\,\, \colon\,\, i \in [1,\ell] \},\ |ID(w_1)| \big\}$,
where $x$ denotes the minimum number of agents present at any node in $G'$, including nodes that contain only inactive agents (i.e., $w_1$). Agent $\mathcal{A}$ selects a node $u$ containing $x$ agents; if multiple choices exist among the nodes $v_i$ and $w_1$, it selects the node containing the minimum ID agent and stores this identifier in $\mathcal{A}.\text{srcID}$. In this round, agent $\mathcal{A}$ remains at its current position in this round.
\end{itemize}

\vspace{0.15cm}
\noindent \textbf{Case 2 (there is an agent $b_j^i$ with $b_j^i.\text{phase}=1$):} It sets $\mathcal{A}.\text{phase}=1$, $\mathcal{A}.\text{srcID}=b_j^i.\text{srcID}$, and $\mathcal{A}.p=\max\{\mathcal{X},\mathcal{Y}\}$. There are two sub-cases.

\begin{itemize}
    \item \textbf{Sub-case 1 (a port in $G'$ leads to a hole)}: Without loss of generality, one of the port of $v_2$ leads to a hole. If $\mathcal{A}$ is at node $v_2$, it set $\mathcal{A}.\text{EXP}=1$, and moves to the hole. Otherwise, it sets $\mathcal{A}.\text{EXP}=1$, and stays at its position.
    \item \textbf{Sub-case 2 (no port in $G'$ leads to a hole)}: 
If there is no port in $G'$ that leads to a node containing only inactive agents, agent $\mathcal{A}$ updates $\mathcal{A}.\text{EXP}=1$. Otherwise, the behavior of $\mathcal{A}$ depends on the value of $\mathcal{A}.\text{srcID}$. If $\mathcal{A}.\text{srcID}\notin ID(v_i)$ for every $i\in[1,\ell]$ and $\mathcal{A}.\text{srcID}\notin ID(w_1)$, then $\mathcal{A}$ sets $\mathcal{A}.\text{EXP}=1$. If $\mathcal{A}.\text{srcID}\in ID(w_1)$, then $\mathcal{A}$ stays at its current position. Otherwise, without loss of generality, assume that $\mathcal{A}.\text{srcID}\in ID(v_1)$. If $|ID(v_1)|=1$ and there exists a port in $G'$ leading to a node containing only inactive agents, then $\mathcal{A}$ sets $\mathcal{A}.\text{EXP}=1$. If this is not the case, then there exists at least one agent at node $v_1$ other than $\mathcal{A}.\text{srcID}$. If no such agent is active, then $\mathcal{A}$ stays at its position. Otherwise, there exists at least one active agent other than $\mathcal{A}.\text{srcID}$, and $\mathcal{A}$ considers a shortest path $P$ between $v_1$ and $w_1$. If multiple shortest paths exist, $\mathcal{A}$ selects the lexicographically smallest one. Let $P=\overline{u}_1(=v_1)\sim\overline{u}_2\sim\cdots\sim\overline{u}_{\lambda}(=w_1)$ denote this path in $G'$. If $\mathcal{A}.\text{ID}=\mathcal{A}.\text{srcID}$, then $\mathcal{A}$ stays at its position; otherwise, if $\mathcal{A}$ is at node $\overline{u}_j$ with $j<\lambda$, it moves to node $\overline{u}_{j+1}$.
\end{itemize}

\subsection{Correctness and analysis of algorithm}
In this section, we prove that {\ssync}\_$\mathcal{EXPO}$ solves exploration under the {\ssync} scheduler with move complexity $O(k\hat{D})$ and memory requirement $O(\max\{\log n, \log p\})$ per agent. Furthermore, every agent eventually becomes aware that exploration has been completed. We have the following observation as per our algorithm. 

\begin{observation}\label{obs:upper_p}
In the algorithm, the value of $\mathcal{A}.p$ for any agent $\mathcal{A}$ increases only if $\mathcal{A}$ receives information from an agent $\mathcal{B}$ with $\mathcal{B}.p > \mathcal{A}.p$, or if $\mathcal{A}$ observes more than $\mathcal{A}.p$ inactive agents in the current round. Since the adversary can deactivate at most $p$ agents in any round, it follows that $\mathcal{A}.p \le p$ for every agent $\mathcal{A}$.
\end{observation}

\begin{lemma}\label{lm:info_spreading}
Let $A_r$ denote the set of agents active in round $r$, and let the adversary deactivate at most $p = \left\lfloor \frac{k-1}{n-2} \right\rfloor - 1$ agents per round. Suppose that at round $r_1$, all agents in $A_{r_1}$ agree on some information $\mathcal{I}$ of size $O(\log n)$. Then, there exists at least one agent that is active in both rounds $r_1$ and $r_1+1$, that is, $A_r \cap A_{r+1} \neq \emptyset $. 
\end{lemma}
\begin{proof}
The adversary can deactivate at most $p = \left\lfloor \frac{k-1}{n-2} \right\rfloor - 1$ agents per round, where $k$ is the total number of agents and $n$ is the number of nodes of $G$. As shown in Section~\ref{sec:imp}, this implies $k \ge (n-2)(p+1)+1$. Consider a round $r_1$ in which at most $p$ agents are inactive; hence, the number of active agents is $|A_{r_1}| = k-p$. Assume that all agents active in round $r_1$ agree on some information $\mathcal{I}$ of size $O(\log n)$. In round $r_1+1$, the adversary may deactivate at most $p$ agents from $A_{r_1}$, so at least $k-2p$ agents retain the information $\mathcal{I}$ at round $r_1+1$. Using the lower bound on $k$, we have $k-2p \ge (n-2)(p+1)+1-2p = np+n-4p-1$. Since $n \ge 4$ (considered in the begging of Section \ref{sec:ssync_algo}), this value is positive, which implies that $A_{r_1} \cap A_{r_1+1} \neq \emptyset$. Therefore, at least one agent is active in both rounds $r_1$ and $r_1+1$. This completes the proof.
\end{proof}
Based on Lemma \ref{lm:info_spreading}, we have the following observation.
\begin{observation}\label{obs:exp}
Let a non-empty set of agents set its $\text{EXP}$ (resp. \text{phase}) parameter to $1$ at round $r_1$. Then, for every round $r \ge r_1+1$, $\exists$ at least one agent $\mathcal{B}$ with $\mathcal{B}.\text{EXP} = 1$ (resp. $\mathcal{B}.\text{phase} = 1$) at round $r$. Also, for every agent $\mathcal{A}$ and every round $r \ge 0$, the value of $\mathcal{A}.p$ is non-decreasing over time.
\end{observation}

\begin{lemma}\label{lm:twohole}
If there are at least two holes in round $r$, then there exists a node that contains at least $p+2$ agents in round $r$. 
\end{lemma}

\begin{proof}
Let $h$ denote the number of holes in round $r$. Since $h \ge 2$, agents can occupy at most $n-2$ nodes. The adversary can deactivate at most $p$ agents per round, where $k$ is the total number of agents and $n$ is the number of nodes of $G$. Suppose, for contradiction, that no node contains $p+2$ agents. Then each occupied node contains at most $p+1$ agents. As agents occupy at most $n-2$ nodes, the total number of agents is at most $(n-2)(p+1)$ which contradicts the fact that $k$ agents can solve the exploration only if $k\geq (n-2)(p+1)+1$ (recall from Section \ref{sec:imp}). Therefore, there exists at least one node that contains at least $p+2$ agents in round $r$. This completes the proof.
\end{proof}

Based on lemma \ref{lm:twohole}, we have the following observation.

\begin{observation}\label{obs:neg2hole}
Let $\mathcal{A}$ be an active agent in round $r$, and suppose that there are at least two holes in $\mathcal{G}_r$. Since $\mathcal{A}.p \le p$ by Observation~\ref{obs:upper_p}, Lemma~\ref{lm:twohole} implies that there exists a node $v$ in $\mathcal{G}_r$ containing at least $\mathcal{A}.p+2$ agents. Consequently, if in round $r$ agents observe that no node contains at least $\mathcal{A}.p+2$ agents, they can correctly conclude that there is at most one hole in $\mathcal{G}_r$.
\end{observation}

We prove that, in round $r$, every active agent possesses complete information about $CCA(\mathcal{G}_r)$.

\begin{lemma}\label{lm:map}
Let $\mathcal{A}$ be an active agent at node $v$ of $CCA(\mathcal{G}_r)$. Then, at the end of Phase 2, agent $\mathcal{A}$ constructs a labeled graph isomorphic to $CCA(\mathcal{G}_r)$. Moreover, for every node in this reconstructed graph, the agent correctly identifies (i) which incident ports lead to holes and (ii) which incident ports lead to nodes occupied exclusively by inactive agents, along with the IDs of such inactive agents. 
\end{lemma}
\begin{proof}
Let $(w_1, w_2)$ be an edge in $CCA(\mathcal{G}_r)$, and let $\pi(w_1, w_2) = q_1$ and $\pi(w_2, w_1) = q_2$. Since $(w_1, w_2)$ belongs to $CCA(\mathcal{G}_r)$, both $w_1$ and $w_2$ contain at least one active agent due to Def. \ref{def:1}. Let $ID_{w_1}$ and $ID_{w_2}$ denote the minimum identifiers among the active agents located at $w_1$ and $w_2$, respectively. Then agent $\mathcal{A}$ receives the 1-hop views $C_{w_1}$ and $C_{w_2}$, and hence $ID_{w_1}, ID_{w_2} \in V'$ as per Phase 2. Because $\pi(w_1, w_2) = q_1$ and $\pi(w_2, w_1) = q_2$, we have that $C_{w_1}^{q_1} \in C_{w_1}$ and $C_{w_2}^{q_2} \in C_{w_2}$. In Phase 2, these tuples are $C_{w_1}^{q_1} = (ID(w_1), q_1, ID(w_2))$ and $C_{w_2}^{q_2} = (ID(w_2), q_2, ID(w_1))$. Therefore, agent $\mathcal{A}$ adds an undirected edge $(ID_{w_1}, ID_{w_2})$ to $E'$ with port labels $\pi(ID_{w_1}, ID_{w_2}) = q_1$ and $\pi(ID_{w_2}, ID_{w_1}) = q_2$. Hence, every edge of $CCA(\mathcal{G}_r)$ is correctly reconstructed.

Now consider a node $w_1 \in CCA(\mathcal{G}_r)$ such that one of its ports, say $q_1$, leads to a node $w_2$ containing only inactive agents. Since $w_1 \in CCA(\mathcal{G}_r)$, it contains at least one active agent, and thus agent $\mathcal{A}$ receives $C_{w_1}$. Since no agent is active at $w_2$, no view $C_{w_2}$ is received. As port $q_1$ of $w_1$ leads to $w_2$, we have $C_{w_1}^{q_1} = (ID(w_1), q_1, ID(w_2))$. By the rules of Phase 2, $ID_{w_1} \in V'$, where $ID_{w_1}$ is the minimum active identifier at node $w_1$, and agent $\mathcal{A}$ correctly records that port $q_1$ of node $ID_{w_1}$ leads to a node with only inactive agents and it marks port $q_1$ with set $ID(w_2)$.

Next, consider a node $w_1 \in CCA(\mathcal{G}_r)$ such that one of its ports, say $q_1$, leads to a hole $w_2$. Since $w_1 \in CCA(\mathcal{G}_r)$, it contains at least one active agent, and hence $C_{w_1}$ is received. As no agent is present at $w_2$, no view from $w_2$ is received. Since port $q_1$ of $w_1$ leads to $w_2$, we have $C_{w_1}^{q_1} = (ID(w_1), q_1, \emptyset)$. By the rules of Phase 2, $ID_{w_1} \in V'$, and agent $\mathcal{A}$ correctly records that port $q_1$ of node $ID_{w_1}$ leads to a hole. This completes the proof.
\end{proof}

\begin{lemma}\label{lm:pipeline}
If in round $r$ there exists a node in $\mathcal{G}_r$ that contains at least $\mathcal{A}.p+2$ agents, then at the end of round $r$ either a hole is filled or a node containing only inactive agents receives at least one active agent, unless $\mathcal{G}_r$ contains no hole and no node consisting only of inactive agents, or there already exists an active agent $\mathcal{B}$ with $\mathcal{B}.\text{EXP} = 1$ (or $\mathcal{B}.\text{phase} = 1$) in round $r$. 
\end{lemma}
\begin{proof}
 As per {\ssync}\_$\mathcal{EXPO}$, the value of $\mathcal{A}.p=\max\{\mathcal{X}, \mathcal{Y}\}$, where $\mathcal{X}$ represents the number of inactive agents at round $r$ which active agents can observe, and $\mathcal{Y}$ is the maximum $p$ value among all active agents at round $r$. Due to Observation \ref{obs:upper_p}, the value of $\mathcal{A}.p$ is the same as $\mathcal{B}.p$ for any two active agents $\mathcal{A}$ and $\mathcal{B}$ at round $r$. Due to Lemma \ref{lm:map}, all active agents form a map of $CCA(\mathcal{G}_r)$ (i.e., $G'$) including information which port(s) lead to a hole or port(s) lead to a node which has only inactive agents. As per {\ssync}\_$\mathcal{EXPO}$, agents do the following. Let $\overline{w}_1, \overline{w}_2, \ldots, \overline{w}_{\lambda_1}$ be the nodes in $G'$ that have at least $\mathcal{A}.p + 2$ agents. Let $\overline{b}_j$ be the minimum ID active agent at node $\overline{w}_j$. Without loss of generality, let $\overline{b}_1.ID = \min\{ \overline{b}_j.ID \mid j \in [1,\lambda_1] \}$. Consider the active connected component (say $H$) of $G'$ where node $\overline{w}_1$ is present.

 If there is a port in $H$ that leads to a hole, then all active agents agree on the path $P$ in $H$ that leads to the hole. Let $P = \overline{u}_1(=\overline{w}_1) \sim \overline{u}_2 \sim \ldots \sim \overline{u}_{\lambda}$ be a path such that one of the ports of node $\overline{u}_{\lambda}$ leads to a hole. Since path $P$ is part of $H\subseteq CCA(\mathcal{G}_r)$, at least one active agent is present at each node $\overline{u}_j$ for every $j\in [1,\lambda]$. As per the algorithmic steps, the minimum active agent at node $\overline{u}_j$ for every $j<\lambda$ moves to node $\overline{u}_{j+1}$. And the minimum active agent at the node $\overline{u}_\lambda$ moves the minimum available port, which leads to a hole. 

 Else, if no port in $H$ leads to a hole, and $\mathcal{G}_r$ contains either a hole or a node consisting only of inactive agents, Definition~\ref{def:1} implies that there exists a port in $H$ leading to a node containing only inactive agents. If there is a port in $H$ that leads to a node with only inactive agents, then all active agents agree on the path $P$ in $H$ that leads to that node as per {\ssync}\_$\mathcal{EXPO}$. Let $P = \overline{u}_1(=\overline{w}_1) \sim \overline{u}_2 \sim \ldots \sim \overline{u}_{\lambda}$ be a path such that one of the ports of node $\overline{u}_{\lambda}$ leads to a node which has only inactive agents. Since path $P$ is part of $H\subseteq CCA(\mathcal{G}_r)$, at least one active agent is present at each node $\overline{u}_j$ for every $j\in [1,\lambda]$. As per the algorithmic steps, the minimum active agent at node $\overline{u}_j$ for every $j<\lambda$ moves to node $\overline{u}_{j+1}$. And the minimum active agent at the node $\overline{u}_\lambda$ moves via the minimum available port, which leads to a node that has only inactive agents present. Therefore, at the end of round $r$, one of the nodes which has only inactive agents gets one agent from one of the nodes of $G'$. This completes the proof.
\end{proof}

We now present a key lemma on which the correctness of the algorithm relies.

\begin{lemma}\label{lm:key}
Let at round $r$ an active agent $\mathcal{A}$ observe that no node in $\mathcal{G}_r$ contains at least $\mathcal{A}.p+2$ agents. If $\mathcal{G}_r$ contains a hole, then there is at most one node in $\mathcal{G}_r$ containing only inactive agents.
\end{lemma}

\begin{proof}
   Let $\beta$ denote the number of nodes in $\mathcal{G}_r$ that contain only inactive agents, and let $w_1, w_2, \ldots, w_\beta$ be these nodes. Since the adversary can deactivate at most $p$ agents in any round, the total number of inactive agents satisfies $\sum_{j=1}^{\beta} ID(w_j) \le p$. As $\beta$ nodes contain only inactive agents and there is exactly one hole in $\mathcal{G}_r$, the number of nodes that contain at least one active agent is $n-\beta-1$. Since no node contains more than $\mathcal{A}.p+1$ agents at round $r$, the total number of agents is at most $(\mathcal{A}.p+1)(n-\beta-1)+p$. Therefore, the following inequality must hold:
\begin{equation}
(\mathcal{A}.p+1)(n-\beta-1)+p \ge (n-2)(p+1)+1 .
\label{eq:algo1}
\end{equation}

\vspace{0.15cm}
\noindent \underline{\textbf{Case 1 $\bm{\mathcal{A}.p = p}$}:}
From Eq.~\ref{eq:algo1}, we obtain $(p+1)(n-\beta-1)+p \ge (n-2)(p+1)+1 \;\Longrightarrow\; 
\beta \le \frac{2p}{p+1} \le 2 .$ Thus, $\beta \le 2$. We now show that $\beta = 2$ is impossible. Substituting $\mathcal{A}.p = p$ and $\beta = 2$ into Eq.~\ref{eq:algo1} gives $(p+1)(n-3)+p \ge (n-2)(p+1)+1 \;\Longrightarrow\; 2 \le 0 $ which is a contradiction. Hence, $\beta \le 1$.

\vspace{0.15cm}
\noindent \underline{\textbf{Case 2 $\bm{\mathcal{A}.p < p}$}:}
From Eq.~\ref{eq:algo1}, we have $(\mathcal{A}.p+1)(n-\beta-1)+p \ge (n-2)(p+1)+1 
\;\Longrightarrow\; 
\mathcal{A}.p (n-\beta-1)-\beta \ge np-3p$. Since $\mathcal{A}.p < p$, it follows that $p(n-\beta-1) -\beta > np-3p 
\;\Longrightarrow\; 
\beta < \frac{2p}{p+1} \;\Longrightarrow\; 
\beta < 2 .$ Therefore, $\beta \le 1$ holds in this case as well. This completes the proof.
\end{proof}

We now state the final result.

\begin{theorem}\label{thm:final}
{\ssync}\_$\mathcal{EXPO}$ solves exploration under the {\ssync} scheduler. 
The move complexity of the algorithm is $O(k\hat{D})$, and each agent requires $O(\max\{\log n, \log p\})$ memory. 
Furthermore, every agent eventually becomes aware that exploration has been achieved.
\end{theorem}

\begin{proof}
Due to Lemma \ref{lm:pipeline}, if in round $r$ there exists a node in $\mathcal{G}_r$ with at least $\mathcal{A}.p+2$ agents, then by the end of round $r$ either a hole is filled or a node containing only inactive agents receives an active agent, unless $\mathcal{G}_r$ contains no hole and no node consisting only of inactive agents, or there already exists an active agent $\mathcal{B}$ with $\mathcal{B}.\text{EXP} = 1$ (or $\mathcal{B}.\text{phase} = 1$) in round $r$. Therefore, in each such round, the number of agents at one of the nodes containing at least $\mathcal{A}.p+2$ agents decrease by at least one. Hence, there exists a round within the first $k$ rounds in which agents observe that every node of $G'$ contains at most $\mathcal{A}.p+1$ agents, where $\mathcal{A}$ is an active agent in that round. \underline{Let $r_1$ denote the first such round}. As per Observation \ref{obs:neg2hole}, there is at most one hole at round $r_1$. There are two possible cases:

\vspace{0.15cm}
\noindent \textbf{Case 1:} If there exists a round $r \le r_1$ in which an active agent observes that no port of $G'$ leads to a hole and no port leads to a node containing only inactive agents, this implies that every node of $G$ has at least one active agent at round $r$. As per {\ssync}\_$\mathcal{EXPO}$, all active agents can conclude correctly at round $r$ that exploration has been completed and set their $\text{EXP}$ parameter to $1$ as per. By Observation~\ref{obs:exp}, at least one agent is active in both rounds $r$ and $r+1$. Without loss of generality, let $\mathcal{A}$ be an agent that is active in both rounds $r$ and $r+1$. Therefore, at the beginning of round $r+1$, there is an agent with its $\text{EXP}$ parameter equal to $1$. Any agent activated in round $r+1$ that was inactive in round $r$ learns this information via global communication and updates its own $\text{EXP}$ parameter to $1$. By the fairness of the {\ssync} scheduler, every agent is eventually activated and hence eventually sets its $\text{EXP}$ parameter to $1$. This step also guarantees that each agent is aware that exploration has been achieved within finite but unbounded time.

\vspace{0.15cm}
\noindent \textbf{Case 2:} As per {\ssync}\_$\mathcal{EXPO}$, each active agent $\mathcal{A}$ at round $r_1$ updates $\mathcal{A}.\text{srcID}$ to the minimum ID of a node that contains the smallest number of agents. This node can also be a node that contains only inactive agents. By Lemma~\ref{lm:info_spreading}, there exists an active agent $\mathcal{B}$ in every round $r_2 \ge r_1$ such that $\mathcal{B}.\text{srcID} \neq \bot$ and $\mathcal{B}.\text{phase}=1$. Any agent $\mathcal{A}_1$ with $\mathcal{A}_1.\text{srcID} = \bot$ and $\mathcal{A}_1.\text{phase} = 0$ updates $\mathcal{A}_1.\text{srcID} = \mathcal{B}.\text{srcID}$ and $\mathcal{A}_1.\text{phase}=1$. Therefore, this source is fixed in every round $r_2 \ge r_1$, and agents execute the algorithm of Case 2 of Phase 3 of {\ssync}\_$\mathcal{EXPO}$ in every $r_2\geq r_1$. Now, we consider the following sub-cases in which exploration can be correctly concluded.

\begin{enumerate}
    \item \textbf{Sub-case 1:} At some round $r_2 \ge r_1$, if active agents find that there are two distinct nodes containing only inactive agents (i.e., $\tau \ge 2$), then every active agent $\mathcal{A}$ sets $\mathcal{A}.\text{EXP}=1$ as per {\ssync}\_$\mathcal{EXPO}$. This correctly implies that exploration is completed at round $r_2$ for the following reason. 
    
    Since active agents observe two distinct nodes containing only inactive agents (i.e., $\tau \ge 2$), no hole is present. Indeed, if a hole were present, then $\tau \le 1$ by Lemma~\ref{lm:key}. Therefore, every active agent $\mathcal{A}$ correctly concludes that exploration is achieved at round $r_2$. As in Case~1, within finite but unbounded time, each agent updates its $\text{EXP}$ parameter to $1$.

    \vspace{0.15cm}
    \item \textbf{Sub-case 2:} At some round $r_2 \ge r_1$, if an active agent $\mathcal{A}$ does not find the node corresponding to $\mathcal{A}.\text{srcID}$ in its 1-hop neighborhood and does not receive this information from the 1-hop views of other active agents, but there exists a port of $G'$ that leads to a node containing only inactive agents, then every active agent $\mathcal{A}$ sets $\mathcal{A}.\text{EXP}=1$ as per {\ssync}\_$\mathcal{EXPO}$. This correctly implies that exploration is completed at round $r_2$ for the following reason. 
    
    If a hole were present, then each active agent would receive information about $\mathcal{A}.\text{srcID}$ in round $r_2$ as $\tau\leq 1$ due to Lemma \ref{lm:key}. Since this does not happen, no hole is present, and hence every active agent correctly concludes that exploration is achieved. As in Case~1, within finite but unbounded time, each agent updates its $\text{EXP}$ parameter to $1$.

    \vspace{0.15cm}
    \item \textbf{Sub-case 3:} At some round $r_2 \ge r_1$, if one of the ports of $G'$ leads to a hole, as per {\ssync}\_$\mathcal{EXPO}$, one active agent moves from a node of $G'$ to the hole, and every active agent $\mathcal{A}$ sets $\mathcal{A}.\text{EXP}=1$ as per {\ssync}\_$\mathcal{EXPO}$. This correctly implies that exploration is completed at round $r_2$ for the following reason. 
    
    Since at round $r_1$, each active agent observes at most $\mathcal{A}.p+1$ agents, there is at most one hole due to Observation \ref{obs:neg2hole}. As per {\ssync}\_$\mathcal{EXPO}$, in subsequent rounds $r' \geq r_1+1$, movement occurs only when a node containing only inactive agents is observed by active agents, and node $v$ with $\mathcal{A}.\text{srcID}$ has at least one active agent other than $\mathcal{A}.\text{srcID}$. Therefore, no new hole is created. If round $r_2$ is the first round where one of the ports of $G'$ leads to a hole, the hole is visited by at least one agent in round $r_2$. Therefore, every active agent correctly concludes that exploration is achieved. As in Case~1, within finite but unbounded time, each agent updates its $\text{EXP}$ parameter to $1$.

    \vspace{0.15cm}
    \item \textbf{Sub-case 4:} Let $v$ be the node where $\mathcal{A}.\text{srcID}$ is present. Consider any round $r_2 \ge r_1$ such that node $v$ contains only the agent with $\mathcal{A}.\text{srcID}$, this agent is active, and there exists at least one node containing only inactive agents. In this situation, every active agent $\mathcal{A}$ sets $\mathcal{A}.\text{EXP}=1$ according to {\ssync}\_$\mathcal{EXPO}$, and exploration is correctly concluded due to the following reasons.
    
    At round $r_1$, active agents observe that at most $\mathcal{A}.p+1$ agents are present at any node in $G'$. Hence, no node in $\mathcal{G}_{r_1}$ contains at least $p+2$ agents as $\mathcal{A}.p\leq p$. As per {\ssync}\_$\mathcal{EXPO}$, in every round $r' \geq r_1+1$, movement occurs only when a node containing only inactive agents is observed by active agents, and node $v$ with $\mathcal{A}.\text{srcID}$ has at least one active agent other than $\mathcal{A}.\text{srcID}$, irrespective of $\mathcal{A}.\text{srcID}$ is active or inactive. Therefore, every node in $\mathcal{G}_{r'}$ contains at most $p+1$ agents at the end of round $r'$. If a hole exists at round $r_2$, no node contains at least $p+2$ agents, and one node $v$ has exactly one active agent ($\mathcal{A}.\text{srcID}$), then at round $r_2$, agents are in $\mathcal{C}^*$ configuration. At round $r_2$, active agents observe that there exists a node containing only inactive agents other than node $v$ (recall node $v$ has exactly one agent, and it is active). From this, they infer that the system is not in configuration $\mathcal{C}^*$ and that no hole exists at round $r_2$. Consequently, every active agent at round $r_2$ correctly concludes that exploration has been achieved. As in Case~1, within finite but unbounded time, every agent updates its $\text{EXP}$ parameter to $1$.
\end{enumerate}

Let $v$ be the node where $\mathcal{A}.\text{srcID}$ is located at round $r_1$. Since at round $r_1$ the active agent $\mathcal{A}$ selects $\mathcal{A}.\text{srcID}$ based on a node that contains the minimum number of agents or a node containing only inactive agents, the number of agents at node $v$ is at most $p$. As per {\ssync}\_$\mathcal{EXPO}$, whenever an agent at node $v$ is active other than $\mathcal{A}.\text{srcID}$, all active agents compute a shortest path $P$ between $v$ and $w_1$ (where $w_1$ denotes a node containing only inactive agents). Using a pipeline strategy along $P$, node $w_1$ receives one agent. If Sub-cases~1,~2, and~3 do not occur within at most $p$ such a pipeline from node $v$, then Sub-case~4 must occur. Therefore, from round $r_1$ onward, there exists a round $r_2$ within finite but unbounded time in which one of the four sub-cases applies.

The number of moves after which at least one active agent $\mathcal{A}$ sets $\mathcal{A}.\text{EXP}=1$ is at most $O(k\hat{D})$. This bound follows from the following reasons. At most $k$ pipelines are required for the agents to reach a configuration in which active agents observe that each node contains at most $\mathcal{A}.p+1$ agents, and thereafter at most $p= \left\lfloor \frac{k-1}{n-2} \right\rfloor - 1\leq k$ additional pipeline can occur from the node $v$ where $\mathcal{A}.\text{srcID}$ is present. As node $v$ cannot contain at most $p$ agents. Each pipeline length cannot be more than $\hat{D}$. Therefore, the move complexity is $O(k\hat{D})$.

Each agent $\mathcal{A}$ maintains the variables $\mathcal{A}.\text{ID}$, $\mathcal{A}.p$, $\mathcal{A}.\text{EXP}$, $\mathcal{A}.\text{srcID}$, and $\mathcal{A}.\text{phase}$. By definition, the parameters $\mathcal{A}.\text{EXP}$ and $\mathcal{A}.\text{phase}$ require $O(1)$ memory. The parameters $\mathcal{A}.\text{ID}$ and $\mathcal{A}.\text{srcID}$ store agent identifiers and therefore require $O(\log n)$ bits as per our model. By Observation~\ref{obs:upper_p}, we have $\mathcal{A}.p \le p$, and hence storing $\mathcal{A}.p$ requires $O(\log p)$ bits. Since $p = \left\lfloor \frac{k-1}{n-2} \right\rfloor - 1 \le k$, this requires $O(\log \frac{k}{n})$ bits. Therefore, the total memory required by each agent is $O(\max(\log n, \log p))$.

 This completes the proof.
\end{proof}

\section{Conclusion}\label{sec:con}
We studied exploration in $1$-interval connected graphs under the {\ssync} scheduler with the dynamic port labeling. We established a bound on the adversary’s deactivation power and presented an exploration algorithm that matches this bound. We have identified the minimal visibility and communication assumptions required for solvability; however, our algorithm relies on global communication. Though we could not formally prove global communication as a necessary condition, the following challenges appear if one considers even $(\hat{D}-1)$-hop communication and $1$-hop visibility. An adversary can exploit dynamic port labeling and selective activation to maintain indistinguishable local views across agents, thereby disrupting coordination. Hence, the pipeline strategy does not work properly and may create new hole(s). This highlights the technical challenge of maintaining coordinated progress under the dynamic port labeling and partial activation. In~\cite{saxena_2025,Saxena_DISC,Ajay_dynamicdisp}, the global communication is assumed to establish feasibility under the {\fsync} scheduler. This leads to two questions: whether global communication is necessary for solvability, and for which values of $k$ exploration can be achieved without additional assumptions.

\bibliography{bib}
\end{document}